\documentclass[manuscript,screen,authorversion,nonacm,preprint]{jair}

\setcopyright{cc}
\setcopyright{none}
\acmDOI{10.1613/jair.1.xxxxx}

\JAIRAE{Insert JAIR AE Name}
\JAIRTrack{} %
\acmVolume{4}
\acmArticle{6}
\acmMonth{8}
\acmYear{2026}

\RequirePackage[
  datamodel=acmdatamodel,
  style=acmauthoryear,
  backend=biber,
  giveninits=true,
  uniquename=init
  ]{biblatex}

\usepackage{multicol}

\usepackage{xspace}

\usepackage{amsmath}
\usepackage{amsthm}
\usepackage{nicefrac}
\usepackage{thm-restate}
\usepackage{mathtools}
\usepackage{bm}

\renewcommand\tableofcontents{\listoftoc*{toc}} %

\usepackage[dvipsnames]{xcolor}
\usepackage{graphicx}
\usepackage{tikz}
\usetikzlibrary{backgrounds,calc}

\usepackage{hyperref}
\hypersetup{
    pdfencoding=auto, 
    psdextra,
    colorlinks=true,
    citecolor=green!40!black,
    linkcolor=red!50!black,
    urlcolor=blue!80!black
}
\usepackage[nameinlink]{cleveref}
\usepackage{doi}

\usepackage{tabularx}
\usepackage{booktabs}
\usepackage{multirow}

\usepackage{algorithm}
\usepackage{algorithmic}

\usepackage{pdfpages}
\usepackage{multicol}

\usepackage{comment}
\usepackage{todonotes}
\usepackage{tikz}
\usetikzlibrary{calc,scopes,backgrounds,decorations.pathmorphing,fit,math}
\usetikzlibrary{decorations.markings}

\newcommand{\N}{\mathbb{N}}
\newcommand{\Z}{\mathbb{Z}}

\newcommand{\Yes}{\textsc{Yes}\xspace}
\newcommand{\No}{\textsc{No}\xspace}

\newcommand{\agents}{{\ensuremath{{N}}}}    %
\newcommand{\numAgents}{{\ensuremath{{n}}}} %
\newcommand{\pttn}{{\ensuremath{{\pi}}}}    %
\newcommand{\val}{{\ensuremath{{v}}}}       %
\newcommand{\numCoals}{{\ensuremath{{k}}}}  %
\newcommand{\sVec}{\ensuremath{\vec{\bm{n}}}} %
\newcommand{\lVec}{\operatorname{lb}}
\newcommand{\uVec}{\operatorname{ub}}
\newcommand{\wFn}{{\ensuremath{\omega}}} %

\newcommand{\vc}{\operatorname{vc}} %
\newcommand{\tw}{\operatorname{tw}} %

\newcommand{\probName}[1]{\textsc{#1}\xspace}

\newcommand{\Oh}[1]{{\mathcal{O}\left(#1\right)}}

\NewDocumentCommand{\cc}{ O{} O{} m }{\mbox{%
    \expandafter\ifx\expandafter\relax\detokenize{#2}\relax\else{#2-}\fi%
    \textsf{#3}%
    \expandafter\ifx\expandafter\relax\detokenize{#1}\relax\else{-#1}\fi%
    }\xspace}
\newcommand{\NP}{\cc{NP}}
\newcommand{\NPh}{\cc[hard]{NP}}
\newcommand{\NPhness}{\cc[hardness]{NP}}
\newcommand{\NPc}{\cc[complete]{NP}}

\newcommand{\FPT}{\cc{FPT}}
\newcommand{\XP}{\cc{XP}}
\newcommand{\W}[1][1]{\cc{W[#1]}}
\newcommand{\Wh}[1][1]{\cc[hard]{W[#1]}}
\newcommand{\Whness}[1][1]{\cc[hardness]{W[#1]}}

\newtheorem{claim}{Claim}[section]
\newtheorem{remark}{Remark}[section]
\newenvironment{claimproof}[1]{\textsc{Proof.}\hspace{0.15cm}#1}{\hfill$\blacktriangleleft$\medskip}

\definecolor{cbYellow}{RGB}{238,204,102}
\definecolor{cbBlue}{RGB}{102,153,204}
\definecolor{cbRed}{RGB}{153,68,85}
\definecolor{cbGreen}{RGB}{34,136,51}

\begin{document}

\fancypagestyle{firstpagestyle}{
    \fancyhf{}%
    \fancyfoot[RO,LE]{}%
  }%

\fancyfoot{}

\title[Individual Rationality in Constrained Hedonic Games: 
Additively Separable and Fractional Preferences]{Individual Rationality in Constrained Hedonic Games:\\ 
Additively Separable and Fractional Preferences}
\titlenote{An extended abstract of this work has been published in the Proceedings of the 25th International Conference on Autonomous Agents and Multiagent Systems, {AAMAS}~'26~\citep{FioravantesGMS2026}.}
\author{Foivos Fioravantes}
\affiliation{
    \institution{Czech Technical University in Prague}
    \city{Prague}
    \country{Czech Republic}
}
\email{foivos.fioravantes@fit.cvut.cz}
\orcid{0000-0001-8217-030X}

\author{Harmender Gahlawat}
\affiliation{
    \institution{Ben-Gurion University of the Negev}
    \city{Beersheba}
    \country{Israel}
}
\email{harmendergahlawat@gmail.com}
\orcid{0000-0001-7663-6265}

\author{Nikolaos Melissinos}
\affiliation{
    \department{Computer Science Institute, Faculty of Mathematics and Physics}
    \institution{Charles University}
    \city{Prague}
    \country{Czech Republic}
}
\email{melissinos@iuuk.mff.cuni.cz}
\orcid{0000-0002-0864-9803}

\author{Šimon Schierreich}
\affiliation{
    \institution{AGH University of Krakow}
    \city{Krakow}
    \country{Poland}
}
\affiliation{
    \institution{Czech Technical University in Prague}
    \city{Prague}
    \country{Czech Republic}
}
\email{schiesim@fit.cvut.cz}
\orcid{0000-0001-8901-1942}

\renewcommand{\shortauthors}{Fioravantes, Gahlawat, Melissinos \& Schierreich}

\begin{abstract}
    Hedonic games are an archetypal problem in coalition formation, where a set of selfish agents want to partition themselves into \emph{stable} coalitions. In this work, we focus on two natural constraints on the possible outcomes. First, we require that exactly~$\numCoals$ coalitions are created. Then, loosely following the model of Bilò et al. (AAAI 2022), we assume that each of the~$\numCoals$ coalitions is additionally associated with a lower and upper bound on its size. The notion of stability that we study is that of \emph{individual rationality} (IR), which requires that no agent strictly prefers to be alone compared to being in his or her coalition.
     
    Although IR is trivially satisfiable even in the most general models of hedonic games, the complexity picture of deciding whether an IR allocation exists, considering the above constraints, is unexpectedly rich. We reveal that tractable fragments of this computational problem require surprisingly nontrivial arguments, even if we restrict ourselves to \emph{additively separable} and \emph{fractional hedonic games}. Our tractability results, achieved by exploiting the structure of the underlying \textit{preference graph}, are also complemented by their intractability counterparts, painting a fairly complete picture of the tractability landscape of this problem.
\end{abstract}

\maketitle

\section{Introduction}

Organizing a seating arrangement for a banquet may seem straightforward, but quickly becomes complex when interpersonal relationships are considered. For example, placing Janine from HR at the same table as Sheldon from the Physics Department could result in a very unpleasant evening---not just for them, but likely for all the attendees. And it is not only about animosities; seating Sheldon at the table only with strangers can be a similarly explosive choice. In addition, there may be certain restrictions on seating, including the number of available tables and their capacity.

The situation introduced above appears naturally in various domains and falls in the area of \emph{coalition formation}. More formally, the goal here is to partition a set of agents into several groups, called \emph{coalitions}, so that no agent is motivated to deviate from its coalition to improve their utility. In more game-theoretical terms, we are interested in \emph{stable} partitions. The prevailing model in this area is that of \emph{hedonic games}~\citep{DrezeG1980,AzizS2016}, where agents care only about other agents in their own coalition and not about how the additional coalitions are formed.

Traditionally, arbitrarily constructed coalition structures are acceptable as long as they meet certain stability criteria. However, in many real-life scenarios, such as the one at the beginning of this section, we have an additional restriction on the outcome. First, as argued by \citet{SlessHKW2018}, in certain situations, it may be necessary to form an exact number of coalitions; see also \citep{WaxmanKH2021,BarrTKRH24}. We call this model \emph{$k$-hedonic games}. A different restriction, motivated by the recently proposed model of \citet{BiloMM2022} (see also \citep{LiMNS2023,DeligkasEIKS2025,AgarwalARN2025}), asks for~$k$ coalitions, each of size between a given lower and upper bound. We call this model \emph{hedonic games with size-constrained coalitions}. Several of these works study the computational aspects of finding stable outcomes under these two restrictions. However, they are mostly interested in Nash, individual, or core stability, which are all very demanding, rarely satisfiable, and the associated computational problems are often intractable. In our work, we return to the foundational stability notion of \emph{individual rationality} (IR)~\citep{Roth1977,AshlagiR2011,PolevoyD2022,DeligkasEKS2024b}, which is strictly weaker than the notions mentioned above (see, e.g., \citep[Figure 15.1]{AzizS2016}). Intuitively, a coalition structure is IR if no agent strictly prefers to be alone compared to being in its current coalition. 
It is a well-known fact that every hedonic game admits a trivial individually-rational solution; just place each agent in a coalition on its own. Due to this, individual rationality is completely overlooked in the relevant literature. However, as we show, our restrictions turn IR into an unexpectedly interesting solution concept with a very colorful complexity picture.

\begin{table*}[bt!]
\caption{Summary of our contributions. We use~$\vc(G)$ to denote the vertex cover number of~$G$, and~$\tw(G)$ for its treewidth.}
\label{tab:contributions}
\centering
\renewcommand{\arraystretch}{1.3}
\begin{tabularx}{\linewidth}{@{}lXcc@{}} 
\toprule
\textbf{Problem} & \textbf{Restrictions} & \textbf{Complexity} & \textbf{Reference} \\
\midrule
\multirow{2}{*}{\begin{tabular}[c]{@{}l@{}}Both {$k$-ASHGs} \\and {ASHGs-SCC}
\end{tabular}} 
& for every~$k \geq 2$, symmetric valuations,~$G$ is a split graph with~$\Oh{k^2}$ negative edges and~$\vc(G)=2k$ & \NPc & \Cref{thm:ASHG:VC:NPhard}  \\
\cmidrule(l){2-4}

& for every~$k \geq 2$, symmetric valuations,~$G$ is a clique, and the valuations use six different unary encoded values & \NPc & \Cref{thm:ASHG:NPc:clique} \\
\cmidrule(l){2-4}

&  symmetric, unary encoded valuations, parameterized by the number of coalitions~$k$, the number of negative edges, and~$\vc(G)$ & \Wh & \Cref{thm:ASHG:VC:unary:Whard}\\
\cmidrule(l){2-4}

&  unary encoded valuations, parameterized by~$\vc(G)$ & \XP & \Cref{thm:ASHG:VC:unary:XP}\\
\cmidrule(l){2-4} 

& for~$k=2$, symmetric and binary valuations,  & \Wh & \Cref{thm:ASHG:TD:binary:Whard}\\
& parameterized by the treedepth of~$G$ &  \\

\midrule
\multirow{1}{*}{\begin{tabular}[c]{@{}l@{}} {Only~$k$-ASHGs}
\end{tabular}} 

& parameterized by~$\vc(G)$ and the maximum weight~$\wFn_{\max}$ 
& \FPT & \Cref{thm:ASHG:FPT:VC:Wmax}\\
\cmidrule(l){2-4} 
& binary valuations, parameterized by~$\tw(G)$  & \XP & \Cref{thm:ASHG:TW:binary:XP} \\

\midrule
\multirow{1}{*}{\begin{tabular}[c]{@{}l@{}} {ASHGs-SCC}
\end{tabular}} 

& parameterized by~$\vc(G)$, the number of coalitions $\numCoals$, and the maximum weight~$\wFn_{\max}$ 
& \FPT & \Cref{thm:SCC:FPT:VC:k:Wmax}\\
\cmidrule(l){2-4} 
& binary valuations, parameterized by~$\tw(G)$ and the number of coalitions $\numCoals$  & \XP & \Cref{thm:ASHG:TW:binary:XP} \\

\bottomrule
\end{tabularx}
\end{table*}

\subsection{Our Contribution}
The main contribution of our work is a detailed algorithmic landscape of the two restrictions of coalition formation introduced above assuming two prominent subclasses of hedonic games. As classical complexity renders the problem very soon intractable, we turn to the finer-grained framework of parameterized complexity\footnote{A formal introduction to parameterized complexity is given in \Cref{sec:prelims}.}. Intuitively, we study the complexity of a computational problem not only with respect to the input size, but also assuming some \emph{parameter}~$p$ additionally restricting the input instance. The ultimate goal of parameterized algorithms is to find (or prove that it is not possible by showing that the problem is \Wh) an algorithm whose running time is exponential only in the parameter, and its time dependence on the input size is only polynomial in it. Such algorithms are called \FPT. Slightly worse are \XP algorithms, whose running time is also polynomial in the input size, but the degree of this polynomial depends on the parameter. One can conditionally rule out the existence of an \XP algorithm by showing that a problem is \NPh already for a constant value of~$p$.

The two classes of hedonic games we study are \emph{additively separable hedonic games} (ASHGs)~\citep{BogomolnaiaJ2002} and \emph{(modified) fractional hedonic games} (MFHGs)~\citep{AzizBBHOP2019,Olsen2012}. In these games, each agent~$i$ has a specific value~$\val_i(j)$ for every other agent~$j$. The utility of an agent~$i$ in a coalition~$C$ containing~$i$ is then simply the sum of individual values agent~$i$ has for all other agents in~$C$ (divided by the size of~$C$ in case of FHGs). We can succinctly encode such preferences using a weighted digraph~$G$ over the set of agents containing an arc~$(i,j)$ of weight~$\val_i(j)$ for every pair of distinct agents~$i$ and~$j$ (and we can remove all zero weight arcs). 

First, we observe that if the desired stability notion is individual rationality, these two classes coincide; therefore, it suffices to study only ASHGs. Then, we investigate the complexity of ASHGs with respect to two natural restrictions: the structure of the preference graph~$G$ and the maximum value (arc weight) we have in the instance. As we reveal, these two dimensions of the problem interplay in a non-trivial way. Our results are summarized in \Cref{tab:contributions}.

Specifically, if the weights can be general, both variants of ASHGs are \NPc, even if the graph~$G$ is a split graph of constant vertex cover number. If the weights are polynomially bounded by the number of agents, then there is an \XP algorithm with respect to the vertex cover number of~$G$. For \FPT tractability, we need to additionally parameterize by the maximum weight, as without this both problems are \Wh already for the combined parameter of the number of coalitions and the vertex cover of~$G$. Next, we show that the \XP (or \FPT) algorithm for the vertex cover number parameterization cannot be extended to more general graph classes, as we show that both problems are \Wh when parameterized by treedepth, even if~$\numCoals=2$ and the valuations are binary. Treedepth is a structural parameter that lies between the vertex cover number and the celebrated treewidth. The complexity picture is completed by an \XP algorithm for parameterization by treewidth and binary valuations, and we also show that the problem is highly intractable in \emph{dense} graphs by showing \NPhness of both variants even if the underlying graph is a clique. Our algorithms use various techniques, notably including the N-fold ILP formulation---a technique which is completely novel in the area of coalition formation. In addition, all our hardness results hold for symmetric preferences, while our algorithms work also for non-symmetric ones. This draws our results even stronger.

\subsection{Related Work}

\emph{Additively separable hedonic games} (ASHGs)~\citep{BogomolnaiaJ2002}, \emph{fractional hedonic games} (FHGs)~\citep{AzizBBHOP2019}, and \emph{(modified) fractional hedonic games}~\citep{Olsen2012} are probably the most studied subclass of hedonic games, as is clear from the number of papers investigating them~\citep{SungD2010,AzizBS2013,PetersE2015,Peters2016a,Peters2016b,FlamminiMZ2017,HanakaKL2024,BrandtBT2024,FioravantesGM2025,BullingerCS2025,FanelliMM2025,HanakaIO2025}. However, all of these work focus on more general notions of stability than IR, as IR is trivially satisfiable without any further restriction of the game.

There are also other subclasses of hedonic games, such as anonymous~\citep{BanerjeeKS2001,BogomolnaiaJ2002}, diversity~\citep{BredereckEI2019,BoehmerE2020,GanianHKSS2023}, social distance~\citep{BranzeiL2011,GanianHKRSS2023},~$\mathcal{B}$ games~\citep{CechlarovaR2001,CechlarovaH2003}, or~$\mathcal{W}$ games~\citep{CechlarovaR2001,CechlarovaH2004}. However, in these classes of games, either no coalition is worse than the singleton, and therefore, IR is always achievable even under our constraints, or individual rationality was not explored. The only exception is the recent paper of \citet{DeligkasEKS2024b}, who studied individual rationality in \emph{topological distance games} (TDGs)~\citep{BullingerS2023}. However, the model of TDGs is much more general compared to ours, and therefore none of their (mostly negative) results carry over to our setting.

Next,~$\numCoals$-hedonic games and hedonic games with fixed size coalitions are not the only restrictions studied in the literature. For example, \citet{WrightV2015} introduced hedonic games with restricted maximum coalition size, a variant studied also in \citep{LevingerAH2023,GanianHKSS2023,GanianHKRSS2023,FioravantesGM2025}. However, even under this restriction, individual rationality is easy as we can form an arbitrary number of coalitions of size one. Recently and independently of our work, \cite{BullingerDEG2025} studied a setting of ASHGs where the same lower bound is given for each coalition. Notably, they do not assume IR as a solution concept.

Finally, N-fold integer programming~\citep{LoeraHOW2008,Onn2009,HemmeckeOR2013} is an important technique in the design of parameterized algorithms. It has been used successfully in multiple relevant areas such as scheduling~\citep{KnopK2018,KouteckyZ2025}, voting theory~\citep{KnopKM2020,BlazejKPS2024}, and fair division~\citep{BredereckKKN2019,BredereckFKKN2021,Bredereck0KN2023}, to name at least a few related to computational social choice, algorithmic game theory, and collective decision making. Nevertheless, we are not aware of any work using this technique in the area of coalition formation games.

\section{Preliminaries}\label{sec:prelims}

Let~$i\leq j\in \N$. We use~$[i,j]$ to denote the closed interval~$\{i,\ldots,j\}$, and we set~$[i] = \{1,\ldots,i\}$ and~$[i]_0 = [i] \cup \{0\}$. Additionally, for a set~$S$, we use~$2^S$ to denote the power set of~$S$. 

\subsection{Hedonic Games}
Let~$\agents$ be a finite set of~$\numAgents$ agents. A nonempty subset~$C\subseteq \agents$ is called a \emph{coalition}. When~$|C|=\numAgents$, we say that~$C$ is the \emph{grand coalition}, while if~$|C|=1$, we call~$C$ a \emph{singleton}. A \emph{coalition structure}~$\pttn$ is a partition of agents into coalitions, and by~$\pttn(i)$ we denote the coalition agent~$i\in\agents$ is assigned.%

Each agent~$i\in\agents$ provides its subjective \emph{valuation function}~$\val_i\colon 2^\agents\to\Z$, which assigns to each coalition~$C$ a numerical value expressing~$i$'s satisfaction when being part of~$C$. As is usual in hedonic games, 
we assume that~$\val_i$ is defined only for coalitions that contain the agent~$i$. Slightly abusing the notation, we extend the valuations from coalitions to coalition structures by setting~$\val_i(\pttn) = \val_i(\pttn(i))$ and we say that~$\val_i(\pttn)$ is the \emph{utility} the agent~$i$ gets in the coalition structure~$\pttn$.

In this work, we investigate two different restrictions of hedonic games. In the first variant, we are interested only in coalition structures consisting of exactly~$\numCoals\in\N$ non-empty coalitions. Formally, a \emph{$\numCoals$-hedonic game}~$\Gamma$ is a triple~$(\agents,(\val_i)_{i\in\agents},\numCoals)$, where~$\agents$ is a set of agents,~$\val_i$ is a valuation function for every agent~$i\in\agents$, and~$\numCoals$ is the desired number of coalitions. In the second variant, not only the number of coalitions, but also their sizes are prescribed. A \emph{hedonic game with size-constrained coalitions} (HG-SCC)~$\Gamma$ is a~$5$-tuple~$(\agents,(\val_i)_{i\in\agents},\numCoals,\lVec,\uVec)$, where~$\lVec\colon[\numCoals]\to\N$,~$\uVec\colon[k]\to\N$, and~$\lVec(j) \leq \uVec(j)$ for every~$j\in[\numCoals]$. The goal in HG-SCC is to decide whether a coalition structure~$\pttn=(\pttn_1,\ldots,\pttn_\numCoals)$ exists such that~$\lVec(j) \leq |\pttn_j| \leq \uVec(j)$ for every~$j\in[\numCoals]$. Without loss of generality, we assume that function~$\uVec$ satisfies~$\uVec(1) \geq \uVec(2) \geq \cdots \geq \uVec(\numCoals)$. Additionally, we assume~$\sum_{j\in[\numCoals]}\lVec(j) \leq \numAgents \leq \sum_{j\in[\numCoals]} \uVec(j)$. It is easy to see that~$\numCoals$-HG is a special case of HG-SCC with~$\lVec(j) = 1$ and~$\uVec(j)=\numAgents$ for every~$j\in[\numCoals]$. Therefore, any hardness result for~$\numCoals$-HG implies the same hardness result for HG-SCC, and any algorithmic upper bound for HG-SCC directly carries over to~$\numCoals$-HG.

Throughout the paper, we are interested in finding coalition structures where no agent can improve its utility by \emph{deviating} from the given coalition structure~$\pttn$. If no beneficial deviation is possible, we say that a coalition structure~$\pttn$ is \emph{stable}. Specifically, we adopt the notion of \emph{individual rationality}, which, as argued by \citet{AzizS2016}, is a minimal requirement for a solution to be considered stable and is defined as follows.

\begin{definition}
    A coalition structure~$\pttn$ is \emph{individually rational (IR)} if~$\val_i(\pttn) \geq \val_i(\{i\})$ for every agent~$i\in\agents$.
\end{definition}

Intuitively, individual rationality requires that no agent prefers to deviate from the current partition and act alone. Note that the deviation of agent~$i$ to a singleton coalition may technically violate our condition on having exactly~$\numCoals$ coalitions. Therefore, we interpret such deviations as the deviating agent completely leaving the game rather than actually forming a singleton coalition. Alternatively, one can view such a deviation as a thought experiment: it need not actually occur, but its possibility indicates that the solution is undesirable.

In general, the encoding of the valuation functions may require space exponential in the number of agents---we need to store a value for every agent and every possible coalition. Therefore, we study various restrictions of hedonic games that come with a succinct representation of valuations.

\paragraph{Additively Separable Hedonic Games.} In \emph{additively separable hedonic games} (ASHGs), an agent~$i\in\agents$ has a value~$\val_i(j)$ for each agent~$j\in\agents$. We require~$\val_i(i) = 0$. Then, the utility of agent~$i$ in a coalition structure~$\pttn$ is simply~$\sum_{j\in\pttn(i)} \val_i(j)$. If~$\val_i(j) = \val_j(i)$ for all~$i,j\in\agents$, we say that the valuations are \emph{symmetric}, and if~$\val_i(j)\in\{-1,0,1\}$ for every~$i,j\in\agents$, we call the valuations \emph{binary}. The valuations in ASHGs can be represented using a weighted digraph~$G=(\agents,E,\wFn)$ that contains an edge~$ij$ of weight~$\wFn(ij) = \val_i(j)$ if and only if~$\val_i(j) \not= 0$. When the valuations are symmetric, the graph~$G$ is undirected.

\paragraph{Fractional Hedonic Games.} In \emph{fractional hedonic games} (FHGs), similarly to ASHGs, every agent~$i\in\agents$ has a valuation function~$\val_i$ assigning some value to every agent~$j\in\agents$. These valuations are then extended to coalitions as follows. Let~$C\subseteq \agents$ be a coalition such that~$i\in C$. Then, we have~$\val_i(C) = (\sum_{j\in C} \val_i(j)) / |C|$. Again, we can capture the valuations using a similar weighted digraph as in the case of ASHGs.

\paragraph{Modified Fractional Hedonic Games} \emph{Modified fractional hedonic games} (MFHGs) are defined identically to FHGs, the only difference is that the utility of an agent~$i\in\agents$ in a coalition~$C$ containing~$i$ is~$0$ if~$|\pttn(i)| = 1$ and~$(\sum_{j\in C} \val_i(j)) / (|C|-1)$ otherwise.

\subsection{Parameterized Complexity}
This is a domain of algorithm design in which a finer analysis of computational complexity is proposed, compared to the classic approach. This is done by considering additional measures of complexity, referred to as \textit{parameters}. In short, this field proposes a multidimensional time-complexity analysis, where each parameter defines its own dimension. 
Formally, a parameterized problem is a set of instances~$(x,k) \in \Sigma^* \times \mathbb{N}$;~$k$ is the \textit{parameter}.
In this paradigm, an algorithm is considered efficient if it runs in time~$f(k)|x|^{\Oh{1}}$ time for any arbitrary computable function~$f\colon \mathbb{N}\to\mathbb{N}$. Such algorithms are known as \emph{fixed-parameter tractable} (\FPT). If such an algorithm exists for a problem, we say that this problem \textit{is in \FPT}. On the opposite side of the \FPT class, we have the class of \W[1]-hard problems, and it is widely accepted that if a problem is \W[1]-hard then there exists no \FPT algorithm to solve it. Finally, the middle ground between these notions is captured by the class \XP, which contains all parameterized problems that can be solved in time~$|x|^{f(k)}$ for some computable function~$f$. We refer the interested reader to classical monographs in the field~\citep{downey2012parameterized,CyganFKLMPPS15} for a more in-depth introduction to this topic.

\subsection{Structural Parameters}
Let~$G=(V,E)$ be a graph.
A set~$U\subseteq V$ is a \emph{vertex cover} of~$G$ if for every edge~$e\in E$ it holds that~$U\cap e \not= \emptyset$. The \emph{vertex cover number} of~$G$, denoted~$\vc(G)$, is the minimum size of a vertex cover of~$G$.

\begin{definition}
    A \emph{tree-decomposition} of~$G$ is a pair~$(T,\mathcal{B})$, where~$T$ is a tree,~$\mathcal{B}$ is a family of sets assigning to each node~$t$ of~$T$ its \emph{bag}~$B_t\subseteq V$, and the following conditions hold:
    \begin{itemize}
    	\item for every edge~$uv\in E(G)$, there is a node~$t\in V(T)$ such that~$u,v\in B_t$ and
    	\item for every vertex~$v\in V$, the set of nodes~$t$ with~$v\in B_t$ induces a connected subtree of~$T$.
    \end{itemize}
    The \emph{width} of a tree-decomposition~$(T,\mathcal{B})$ is defined as~$\max_{t\in V(T)} |B_t|-1$, and the treewidth~$\tw(G)$ of a graph~$G$ is the minimum width of a tree-decomposition of~$G$.
\end{definition}

It is well known that computing a tree-decomposition of minimum width is in \FPT w.r.t. the treewidth~\citep{Kloks94,Bodlaender96}. Recently there are even more efficient algorithms that have been proposed for obtaining near-optimal tree-decompositions~\citep{KorhonenL23}.

For algorithmic purposes, we use a slight variation of the above definition of the tree-decomposition which is more suitable for a dynamic programming approach.

\begin{definition}
    A tree-decomposition~$(T,\mathcal{B})$ is called \emph{nice} if every node~$x\in V(T)$ is exactly of one of the following four types:
    \begin{description}
        \item[Leaf Node:]~$x$ is a leaf of~$T$ and~$B_x=\emptyset$. 
        \item[Introduce Node:]~$x$ has a unique child~$y$ and there exists~$v\in V\setminus B_y$ such that~$B_x=B_{y}\cup \{v\}$.
        \item[Forget Node:]~$x$ has a unique child~$y$ and there exists~$v\in B_y$ such that~$B_{x}=B_y\setminus\{v\}$.
        \item[Join Node:]~$x$ has exactly two children~$y$ and~$z$, and~$B_x=B_{y}=B_{z}$.
    \end{description}
\end{definition}
It is known that every graph~$G=(V,E)$ admits a nice tree-decomposition that has width equal to~$\tw(G)$ and such a decomposition can be found efficiently~\citep{B98}.

The \emph{tree-depth} of~$G$ can be defined recursively: if~$|V|=1$ then~$G$ has tree-depth~$1$. Then,~$G$ has tree-depth~$k$ if there exists a vertex~$v\in V$ such that every connected component of~$G[V\setminus\{v\}]$ has tree-depth at most~$k-1$.

\subsection{N-fold ILP}
The goal here is to minimize a linear objective~$f$ over a set of structured constraints. Formally, let~$r,D\in\N$ and~$s_i,t_i\in \N$ for every~$i\in[D]$. An N-fold ILP contains~$d = \sum_{i\in[D]} t_i$ variables partitioned into~$D$ \emph{bricks}. Let~$x^{(i)}$ denote the~$i$-th brick. Then, the constraints of an N-fold ILP have the following form
\begin{align}
    && E_1x^{(1)} + \cdots + E_Dx^{(D)} &= \vec{b}_0 && \label{def:Nfold:global} \\
    \forall i \in [D] && A_i x^{(i)} &= \vec{b}_i && \label{def:Nfold:local} \\
    \forall i \in [D] && \vec{\ell}_i \leq x^{(i)} &\leq \vec{u}_i \label{def:Nfold:box} &&
\end{align}
where~$E_i \in \Z^{r\times t_i}$,~$A_i \in \Z^{s_i\times t_i}$,~$\vec{b}_0 \in \Z^r$,~$\vec{b}_i \in \Z^{s_i}$, and~$\vec{\ell}_i,\vec{u}_i\in\Z^{t_i}$ for every~$i\in[D]$. We call constraints of type \eqref{def:Nfold:global} \emph{global}, constraints of type \eqref{def:Nfold:local} \emph{local}, and constraints of type \eqref{def:Nfold:box} \emph{box}. Observe that~$r$ is the total number of global constraints and we use~$s = \max_{i\in[D]} s_i$ to denote the maximum number of occurrences of a single variable in local constraints. In our algorithms, we use the following result of \citet{EisenbrandHKKLO2025}.%

\begin{theorem}[\citet{EisenbrandHKKLO2025}]\label{thm:Nfold:runningTime}
    An instance of N-fold ILP can be solved in~$a^\Oh{r^2s+rs^2}\cdot d\cdot\log(d)\cdot L$, where~$a = \max_{i\in[D]}\{2,||E_i||_\infty,||A_i||_\infty\}$ and~$L$ is the maximum feasible value of the objective.
\end{theorem}

\section{Individual Rationality Can Be Hard}

Before we start our algorithmic journey, we observe that all classes with graph-restricted preference we assume in this work coincide, if we are interested in IR solutions. Indeed, the only difference between ASHGs and (modified) FHGs is that in the latter class, the additive utility of an agent is normalized with respect to the coalition size. This, however, has no effect on the non-negativity of the utility, and therefore, we obtain the following.

\begin{lemma}
    A (modified) fractional hedonic game~$\Gamma = (\agents,(\val_i)_{i\in\agents},\numCoals,\lVec,\uVec)$ admits an IR coalition structure if and only if an additively separable hedonic game~$\Gamma' = (\agents,(\val_i)_{i\in\agents},\numCoals,\lVec,\uVec)$ admits an IR coalition structure.
\end{lemma}
\begin{proof}
    Let~$\pttn$ be an individually rational coalition structure for~$\Gamma$. Then, since no coalition is of negative size, for every~$i\in\agents$, we have~$\sum_{j\in\pttn_i} \val_i(j) \geq 0$. Therefore,~$\pttn$ is also individually rational for the game~$\Gamma'$. By analogous argumentation, any solution for~$\Gamma'$ is also a solution for~$\Gamma$.
\end{proof}

Therefore, the complexity picture for FHGs is the same as that for ASHGs. For the rest of this paper, we assume only the case of additively separable hedonic games, but the results directly carry over to the setting of FHGs.

First, we show that if we do not have any restriction on valuation functions, deciding whether an individually rational coalition structure exists is computationally highly intractable. 

In particular, the result shows \NPhness already if we ask for partitioning into two coalitions, and the graph is a split graph with four negative edges and a vertex cover of size four.

\begin{theorem}\label{thm:ASHG:VC:NPhard}
    For every~$\numCoals \geq 2$, it is \NPc to decide whether a given~$\numCoals$-ASHG or ASHG-SCC~$\Gamma$ admits an individually rational coalition structure, even if the valuations are symmetric and~$G$ is a split graph with~$\Oh{\numCoals^2}$ negative edges and of vertex cover number~$2\numCoals$.
\end{theorem}
\begin{proof}
    We show \NPhness by a reduction from the \probName{Equitable Partition} problem. In this problem, we are given a multi-set~$A=\{a_1,\ldots,a_{2\ell}\}$ of integers such that~$\sum_{a\in A} a = 2t$, and the goal is to decide whether~$S\subseteq[2\ell]$,~$|S| = \ell$, exists so that~$\sum_{i\in S} a_i = \sum_{i \in [2\ell]\setminus S} a_i = t$. The problem is known to be (weakly) \NPh even if each subset~$X\subseteq A$ of size at most~$\ell-1$ sums up to at most~$t-1$~\citep{DeligkasEKS2024}.

    Let~$\mathcal{I}$ be an instance of the \probName{Equitable Partition} problem. We construct an equivalent instance~$\mathcal{J} = (\agents,(\val_i)_{i\in\agents},\numCoals)$ as follows. We describe the construction using the underlying graph. First, we create a clique with four \emph{set agents}~$x_S$,~$y_S$,~$x_{\lnot{S}}$, and~$y_{\lnot S}$. The edges~$x_Sy_S$ and~$x_{\lnot S}y_{\lnot S}$ are of weight~$-t$, and all other edges are of weight~$-3t$. Next, for every~$a_i\in A$, we introduce an \emph{item agent}~$v_i$ and connect it with an edge of weight~$a_i$ to \emph{all} set agents. Finally, we set~$\numCoals = 2$. It is easy to see that the graph~$G$ is a split graph with~$4$ negative edges and the vertex cover number~$4$.

    For correctness, let~$\mathcal I$ be a \Yes-instance and~$S$ be a sought solution. We set~$\pttn_1 = \{ x_S, y_S \} \cup \{ v_i \mid i\in S \}$ and~$\pttn_2 = \agents\setminus \pttn_1$. Since the item agents have only positive edges, any coalition structure is individually rational for them. For~$x_S$, we have~$\val_{x_S}(\pttn) = \wFn(x_Sy_S) + \sum_{i\in S} \wFn(x_Sv_i) = -t + \sum_{i\in S} a_i = -t + t = 0$. The same argument holds for all the remaining set agents. That is,~$\pttn=(\pttn_1,\pttn_2)$ is an individually rational coalition structure.

    In the opposite direction, let~$\mathcal J$ be a \Yes-instance and~$\pttn$ be an individually rational coalition structure. First, we observe that~$x_S$ is not in the same coalition as~$x_{\lnot S}$ and~$y_{\lnot S}$ and, symmetrically,~$x_{\lnot S}$ is not in the same coalition as~$x_{S}$ and~$y_{S}$. If this were the case, then the utility of any set agent in its coalition would be at most~$-3t + 2t = -t$, which contradicts that~$\pttn$ is individually rational. Hence, without loss of generality,~$\{x_S,y_S\}\subseteq \pttn_1$ and~$\{x_{\lnot S},y_{\lnot S}\} \subseteq \pttn_2$. First, we show that~$|\pttn_1|=|\pttn_2| = \ell+2$. For the sake of contradiction, assume that it is not the case. Then, without loss of generality,~$|\pttn_1| \leq \ell-1$ and consequently,~$\pttn_1$ contains at most~$\ell-1$ item agents. Thus, the utility of~$x_S$ is~$-t + \sum_{v_i\in \pttn_1} a_i = -t + t - \epsilon < 0$, where the~$-\epsilon$ part follows by our assumption that the sum of elements of every subset of~$A$ of size at most~$\ell-1$ is less than~$t$. This contradicts that~$\pttn$ is individually rational, so it must hold that~$|\pttn_1| = |\pttn_2| = \ell+2$. We set~$S = \{ i\in[2\ell] \mid v_i \in \pttn_1 \}$ and claim that~$S$ is a solution for~$\mathcal I$. Indeed, the size of~$S$ is exactly~$\ell$. It remains to show that~$\sum_{i\in S} a_i = t$. If~$\sum_{i\in S} a_i < t$, then~$\sum_{v_i\in \pttn_1} \wFn(x_Sv_i) < t$, meaning that~$\pttn$ is not individually rational for~$x_S$. Thus,~$\sum_{i\in S} a_i \geq t$. If~$\sum_{i\in S} a_i > t$, then~$\sum_{v_i\in\pttn_2} \wFn(w_{\lnot S}v_i) = 2t - \sum_{v_i\in \pttn_1} \wFn(x_Sv_i) \leq t - 1$, which implies that~$\pttn$ is not IR for~$x_{\lnot S}$. Therefore, it has to be the case that~$\sum_{i\in S} a_i = t$, which shows that~$S$ is indeed a solution for~$\mathcal I$.
\end{proof}

Our next result again shows intractability of deciding whether an instance admits an individually rational outcome. However, this time, we focus on instances where agents have complete preferences, i.e., each agent has non-zero utility from any other agent.

\begin{theorem}\label{thm:ASHG:NPc:clique}
    For every~$\numCoals \geq 2$, it is \NPc to decide whether a~$\numCoals$-ASHG or ASHG-SCC~$\Gamma$ admits an individually rational coalition structure, even if the valuations are symmetric, the underlying graph is a clique, and the valuations use six different unary encoded values.
\end{theorem}
\begin{proof}
    We will provide a reduction from the~$t$-\textsc{Clique} problem, where given an instance~$(H,t)$, where~$H$ is a graph and~$t\in \mathbb{N}$, the goal is to decide if~$H$ contains a clique on~$t$ vertices. This problem is well known to be \NP-complete~\citep{GareyJ1979}.%
    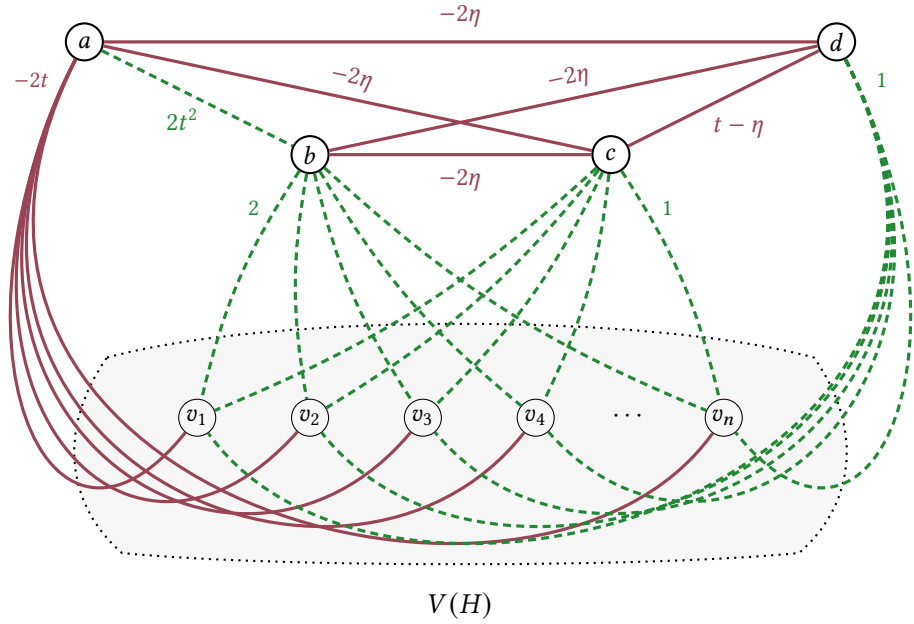
\begin{figure}
        \centering
        \begin{tikzpicture}[
            every node/.style={draw, circle, inner sep=0pt, minimum width=14pt},
            elabel/.style={draw=none, fill=white, inner sep=1.5pt, font=\small},
            redge/.style={very thick,cbRed},
            gedge/.style={very thick, densely dashed,cbGreen},
          ]
        
          \node[thick] (a) at (0, 0)    {$a$};
          \node[thick] (b) at (3, -1.5) {$b$};
          \node[thick] (c) at (7, -1.5) {$c$};
          \node[thick] (d) at (10, 0)   {$d$};
        
          \draw[ultra thick,redge]
              (a) -- (c) node[midway, above, sloped, draw=none,yshift=-2pt] {$-2\eta$}
              (b) -- (c) node[midway, below, draw=none] {$-2\eta$}
              (b) -- (d) node[midway, above, sloped, draw=none,yshift=-2pt] {$-2\eta$}
              (a) -- (d) node[midway, above, draw=none] {$-2\eta$};
          \draw[gedge] (a) -- (b) node[midway, draw=none, below, xshift=-0.2cm] {$2t^2$};
          \draw[redge] (c) -- (d) node[midway, draw=none, below, xshift=0.2cm] {$t-\eta$};
        
          \node (v1) at (1.5, -5)    {$v_1$};
          \node (v2) at (3, -5)   {$v_2$};
          \node (v3) at (4.5, -5)  {$v_3$};
          \node (v4) at (6, -5)    {$v_4$};
          \node[draw=none] (vdots) at (7.25, -5) {$\cdots$};
          \node (vn) at (8.5, -5)    {$v_n$};
        
          \begin{scope}[on background layer]
            \draw[fill=gray!7,thick,dotted]
              (0.3, -4.2) .. controls (2.5,-3.6) and (7.5,-3.6) .. (9.7, -4.2)
              .. controls (10.4,-5.2) and (10.2,-6.2) .. (9.5, -6.8)
              .. controls (7.5,-7) and (2.5,-7) .. (0.5, -6.8)
              .. controls (-0.2,-6.2) and (-0.4,-5.2) .. cycle;
            \node[draw=none, font=\large] at (5, -7.5) {$V(H)$};
          \end{scope}
        
          \foreach \i in {1,2,3,4,n} {
            \draw[bend right=10,gedge] (b) to (v\i);
            \draw[bend left=10,gedge] (c) to (v\i);
          }
          \node[draw=none, font=\small,gedge] at (2.25, -2.2) {$2$};
          \node[draw=none, font=\small,gedge] at (7.75, -2.2) {$1$};
        
          \draw[redge] (a) .. controls (-2.0,-4.0) and (-0.5,-7.5) .. (v1);
          \draw[redge] (a) .. controls (-2.3,-4.5) and ( 0.5,-7.8) .. (v2);
          \draw[redge] (a) .. controls (-2.6,-5.0) and ( 2.0,-8.0) .. (v3);
          \draw[redge] (a) .. controls (-2.8,-5.5) and ( 3.5,-8.2) .. (v4);
          \draw[redge] (a) .. controls (-3.0,-6.0) and ( 6.0,-8.5) .. (vn);
          \node[draw=none, font=\small,redge] at (-0.7, -0.5) {$-2t$};
        
          \draw[gedge] (d) .. controls (12.0,-4.0) and (10.5,-7.5) .. (vn);
          \draw[gedge] (d) .. controls (12.3,-4.5) and ( 8.5,-7.8) .. (v4);
          \draw[gedge] (d) .. controls (12.6,-5.0) and ( 7.0,-8.0) .. (v3);
          \draw[gedge] (d) .. controls (12.8,-5.5) and ( 5.5,-8.2) .. (v2);
          \draw[gedge] (d) .. controls (13.0,-6.0) and ( 4.0,-8.5) .. (v1);
          \node[draw=none, font=\small,gedge] at (10.6, -0.5) {$1$};
        \end{tikzpicture}
        \caption{An illustration of the construction used to prove \Cref{thm:ASHG:NPc:clique}. All green (dashed) edges have positive weight and red (solid) edges have negative weight. Concrete values are depicted next to them. Every pair of $v_i$ and $v_j$ is connected by an edge of weight two if $\{v_i,v_j\}\in E(H)$ and of weight one otherwise.}
        \label{fig:clique}
    \end{figure}

    \smallskip
    \noindent\textbf{Construction.}\hspace{0.25cm} See \Cref{fig:clique} for an illustration. Let~$\eta = |V(H)|$.  We construct a complete graph~$G$ with~$V(G) = V(H)\cup \{a,b,c,d\}$. Moreover, for~$u,v\in V(H)$,  if~$uv \in E(H)$ then~$\wFn(uv) = 2$, else~$\wFn(uv) = 1$. For each~$u\in V(H)$,~$\wFn(av) = -2t$,~$\wFn(bv) = 2$,~$\wFn(cv) = \wFn(dv) = 1$. Next, we set~$\wFn(ab) = 2t^2, \wFn(cd) =t-\eta$, and~$\wFn(ac) = \wFn(ad) = \wFn(bc) = \wFn(bd)= -2\eta$. We remark here that we can safely assume~$t-\eta <0$, as otherwise we can solve our instance of~$t$-\textsc{Clique} in polynomial time. Finally, we have~$k =2$.  This completes our reduction. 

    \smallskip
    \noindent\textbf{Correcteness.}\hspace{0.25cm}
    In one direction, we prove that if~$H$ has a clique on~$t$ vertices, then~$G$ admits an IR coalition structure~$\pi=(\pi_1,\pi_2)$. Let~$U \subset V(H)$ induce a clique on~$t$ vertices. We set~$\pi_1 = U \cup \{a,b\}$ and~$\pi_2 = V(G) \setminus \pi_1$. First, we show that for each~$u\in \pi_1$,~$\val_u(\pi_1) \geq 0$. In particular,~$\val_a(\pi_1) = 0$,~$\val_b(\pi_1) = 2t^2+2t$, and for each~$u\in \pi_i\setminus \{a,b\}$,~$\val_u(\pi_1)= 0$. Second, to see that for each~$u\in \pi_2$,~$\val_u(\pi_2)\geq 0$ observe that for each vertex~$w\in \pi_2\setminus \{c,d\}$, only positive edges are incident to~$w$ in~$G[\pi_2]$ and~$\val_c(\pi_2) = \val_d(\pi_2) = 0$. Thus,~$\pi=(\pi_1,\pi_2)$ is IR stable.

    In the other direction, let~$\pi=(\pi_1,\pi_2)$ be an IR stable solution for~$G$. Firstly, let us assume that~$a\in \pi_1$.  We begin by observing some necessary conditions on~$\pi$ in the following claims.
    \begin{claim}\label{C:C1}
       ~$b\in \pi_1$ and~$c,d\in \pi_2$.
    \end{claim}
    \begin{claimproof}
        Observe that~$c$ and~$d$ value positively only vertices that correspond to vertices of~$H$ and each such valuation is~$1$. Thus, the total amount of positive valuation any of~$c,d$ can obtain in any coalition is~$+\eta$, and since~$\wFn(ac) = \wFn(ad) = \wFn(bc) = \wFn(bd) = -2\eta$, any of~$c,d$ cannot be in the same coalition as~$a,b$ while maintaining individual rationality. Thus, clearly~$c,d\in \pi_2$, and hence,~$b\in \pi_1$ (since~$b$ cannot be in the same coalition as~$c,d$).
    \end{claimproof}

    \begin{claim}\label{C:Csize}
       ~$|\pi_2| \geq \eta-t+2$.
    \end{claim}
    \begin{claimproof}
        Since~$\{c,d\} \subseteq \pi_2$, and for each vertex~$w\in \pi_2\setminus \{d\}$, $\val_c(w)=1$, and~$\val_c(d) = t-\eta$, the total number of vertices in~$\pi_2$ to ensure the individual stability of~$c$ is at least~$\eta-t+2$: $c,d$ and~$\eta-2$ agents corresponding to vertices of~$V(H)$.
    \end{claimproof}

    Now consider~$\pi_1$. Due to \Cref{C:Csize},~$|\pi_1| \leq t+2$. Let~$U = \pi_1\setminus \{a,b\}$. We establish in the following that~$|U| = t$ and~$U$ induces a clique in~$H$ to complete our proof. Clearly~$|U| \leq t$. Now, for any vertex~$u\in U$, we have that~$\val_u(\pi_1)=\val_u(a)+ \val_u(b) + \val_u(U\setminus \{u\})= -2t+2+ \val_u(U\setminus \{u\})$. Since~$\pi$ is individual rational, we have that~$\val_u(U\setminus \{u\}) \geq 2t-2$. Since each vertex in~$U\setminus \{u\}$ is coming from~$H$, for any edge~$uw$, where~$w\in U\setminus \{u\}$,~$\wFn(uw)\leq 2$, we have that~$|U\setminus \{u\}|\geq t-1$ and each~$w$ is adjacent to~$u$ in~$H$ (as only then~$\wFn(uw)\leq 2$). Thus~$|U|=t$ and~$H$ induces a clique in~$H$. This completes the proof for the hardness of ASHG for cliques, even if the valuations are symmetric,~$\numCoals=2$, and the valuations are encoded in unary.
\end{proof}

\section{Preferences with Small Vertex Cover}

The proof of \Cref{thm:ASHG:VC:NPhard} relies on valuations that are exponential in the number of agents. If we restrict the valuations to be polynomially bounded by the number of agents, we obtain a polynomial-time algorithm for every graph with constant size vertex cover. The algorithm is based on dynamic programming over the agents outside of the vertex cover.

\begin{theorem}\label{thm:ASHG:VC:unary:XP}
    If the weights are encoded in unary, there is an algorithm running in~$n^\Oh{\vc(G)}$ time, where~$\vc(G)$ is the vertex cover number of~$G$, that decides whether a~$\numCoals$-ASHG or ASHG-SCC~$\Gamma$ admits an IR coalition structure.
\end{theorem}
\begin{proof}
    \newcommand{\DP}{\operatorname{DP}}
    \newcommand{\uu}{\vec{\mathbf{u}}}
    \newcommand{\nVec}{\vec{\mathbf{s}}} 
    Let~$C$ be a vertex cover of~$G$,~$I = V(G)\setminus C$, and~$\vartheta = |C|$. First, assume that~$\numCoals \leq \vartheta$. In this case, we first guess (by guessing we mean exhaustively trying all possible solutions) a partitioning~$\pttn_C$ of the vertices of~$C$ in a hypothetical solution. Then, we verify that our guess is correct. To do so, we use a dynamic-programming sub-procedure. Specifically, we fix an arbitrary ordering~$v_1,\ldots,v_{n-\vartheta}$ of vertices of~$I$ and create a dynamic-programming table~$\DP[i,\nVec,\uu]$, where~$i\in[n-\vartheta]_0$,~$\nVec=(s_1,\ldots,s_\numCoals)$ is the number of agents in each coalition, and~$\uu = (u_1,\ldots,u_\vartheta)$ is a vector representing the current utility of each agent in the vertex cover. The table stores \texttt{true} if it is possible to extend the partition~$\pttn_C$ with vertices~$v_1,\ldots,v_i$ such that the resulting partition~$\pttn^{i,\nVec,\uu}$ is (a) IR for every agent~$v_1,\ldots,v_i$, (b) the size of every coalition~$\pttn^{i,\nVec,\uu}_j$ in~$\pttn^{i,\nVec,\uu}$ is exactly~$s_j$, and (c)~$\val_{w_j}(\pttn^{i,\nVec,\uu}) = u_j$ for every~$w_j\in C$. Otherwise,~$\DP$ stores \texttt{false}.

    The computation is then defined as follows. For~$i=0$, which represents an auxiliary basic step when no vertex of~$I$ is used, we have
    \begin{equation*}
        \DP[0,\nVec,\uu] = \texttt{true} \iff \nVec = (|\pttn_j^C|)_{j\in[\numCoals]} \land \uu = (\val_{w_i}(\pttn^C))_{i\in[\vartheta]}.
    \end{equation*}
    It is easy to see that for~$i=0$, only one cell is set to \texttt{true}; namely, the one corresponding to the guessed partition of vertex cover agents~$\pttn^C$.

    For every~$i\geq 1$, the computation is then defined as follows. By~$(x,\nVec_{-i})$ we mean a vector created from~$\nVec$ by replacing its~$i$th element with the value~$x$.
    \begin{equation*}
        \DP[i,\nVec,\uu] = \bigvee_{j\in [\numCoals]} \sum_{w \in \pttn^C_j} \val_{v_i}(w) \geq 0\, \land \DP\left[i-1,(s_j - 1,\nVec_{-j}), ( u_\ell - v_{w_\ell}(v_i)\cdot [| w_\ell \in \pttn_j^C |] )_{\ell\in[\vartheta]}\right],
    \end{equation*}
    where~$[|w_\ell \in \pttn_j^C|]$ evaluates to~$1$ if the condition is satisfied and to~$0$ otherwise.

    We prove the correctness using the following two claims. First, we show that whenever~$\DP[i,\nVec,\uu]$ is set to \texttt{true}, there exists a corresponding partial partition~$\pttn^{i,\nVec,\uu}$.

    \begin{claim}
        Let~$(i,\nVec,\uu)$ be a triple such that~$\DP[i,\nVec,\uu] = \texttt{true}$. Then, there exists a partial partition~$\pttn^{i,\nVec,\uu}$ of~$C\cup\{v_1,\ldots,v_i\}$ such that~$\pttn^C_j \subseteq \pttn^{i,\nVec,\uu}_j$ for every~$j\in[\numCoals]$ and satisfying properties (a)-(c).
    \end{claim}
    \begin{claimproof}
        We prove the claim using induction over~$i$. If~$i=0$, then, by the definition of~$\DP$, it must hold that~$\nVec = (|\pttn_j^C|)_{j\in[\numCoals]}$ and~$\uu = (\val_{w_i}(\pttn^C))_{i\in[\vartheta]}$. We set~$\pttn^{i,\nVec,\uu} = \pttn^C$ and claim that it is a valid partial partition. Clearly, as~$i=0$, no agent of~$I$ is expected to be assigned to any of the coalitions, so the property (a) is trivially satisfied. For property (b), the size of every~$\pttn^{i,\nVec,\uu}_j$ is exactly~$|\pttn^C_j| = n_j$. Finally, the utility of every agent~$w_\ell\in C$ is exactly~$\val_{w_\ell}(\pttn^C) = u_\ell$, which again shows that the property (c) is satisfied. Therefore, the computation for~$i=0$ is correct.

        Therefore, let~$i \geq 1$ and assume that~$\DP[i-1,\nVec',\uu']$ is computed correctly for every~$\nVec'$ and~$\uu'$. Then, there exists~$j\in [k]$ such that~$\sum_{w\in \pttn_j^C} \geq 0$ and~$\DP[i-1,(s_j-1,\nVec_{-j}),(u_\ell - v_{w_\ell}(v_i))\cdot [|w_\ell\in\pttn^C_j|])_{\ell\in[\vartheta]}] = \texttt{true}$. Since~$\DP[i-1,(s_j-1,\nVec_{-j}),(u_\ell - v_{w_\ell}(v_i))\cdot [|w_\ell\in\pttn^C_j|])_{\ell\in[\vartheta]}]$ is computed correctly by our assumption, there must exist an extension of~$\pttn^C$ with agents~$v_1,\ldots,v_{i-1}$ which is IR for all these agents, the size of every coalition~$\pttn_\ell$ is~$n_\ell$ with exception of~$\pttn_j$ which is of size~$n-1$, and the utility of every vertex cover agent~$w_\ell\in C$ is~$u_\ell$ if~$w_\ell\not\in \pttn^C_j$ and~$u_\ell - v_{w_\ell}$ if~$w_\ell\in \pttn^C_j$. Let us call this extension~$\pttn'$. We define~$\pttn^{i,\nVec,\uu} = (\pttn'_j\cup\{v_i\},\pttn'_{-j})$ and claim that it is a valid extension with respect to~$(i,\nVec,\uu)$. We changed only the coalition~$j$, so for all agents that are part of the different coalitions and all other coalitions, the properties (a)-(c) are clearly satisfied, as they were satisfied already by~$\pttn'$. Also, by adding an agent~$v\in I$ to a coalition, we cannot break IR for other agents of~$\pttn'_j\cap I$, as they value the agent~$v$ at zero. The same holds for~$v$, which additionally satisfies that the sum of valuations towards all agents in~$\pttn^C_j$ is non-negative. That is, the property (a) is satisfied. The size of~$\pttn^{i,\nVec,\uu}_j$ is~$|\pttn'_j + \{v_i\}| = |\pttn'_j| + 1 = (n_j - 1) + 1 = n_j$, so property (b) is also satisfied. Finally, let~$w_\ell\in\pttn^C$. We have that~$\val_{w_\ell}(\pttn^{i,\nVec,\uu}) = \val_{w_\ell}(\pttn'_j\cup\{v_i\}) = \val_{w_\ell}(\pttn'_j) + \val_{w_\ell}(v_i) = (u_\ell - \val_{w_\ell}(v_i)) + \val_{w_\ell}(v_i) = u_\ell$, finishing the proof.
    \end{claimproof}

    In the opposite direction, we show that whenever for some triple~$(i,\nVec,\uu)$ an extension of~$\pttn^C$ satisfying properties (a)-(c) exists, our dynamic programming table stores in~$\DP[i,\nVec,\uu]$ value \texttt{true}.

    \begin{claim}
        Let~$(i,\nVec,\uu)$ be a triple such that there exists an extension~$\pttn^{i,\nVec,\uu}$ of~$\pttn^C$ such that it satisfies all properties (a)-(c). Then, we have~$\DP[i,\nVec,\uu] = \texttt{true}$. 
    \end{claim}
    \begin{claimproof}
        We again proceed by induction over~$i$. If~$i=0$, no agent~$I$ is assumed, so it must be the case that~$n_j = |\pttn^C_j|$ for every~$j\in[\numCoals]$ and~$u_\ell = \val_{w_\ell}(\pttn^C)$, which is satisfied only by~$\pttn^C$ itself. However, by the definition of the computation, in this case our table stores \texttt{true}, so the basic step is clearly valid.

        Now, let~$i \geq 1$ and assume that the claim holds for~$i-1$. Let~$\pttn^{i,\nVec,\sVec}$ by an extension of~$\pttn^C$ corresponding to~$(i,\nVec,\uu)$ and satisfying properties (a)-(c), and let~$v_i\in\pttn^{i,\nVec,\sVec}_j$ for some~$j\in[\numCoals]$. Since~$\pttn^{i,\nVec,\sVec}$ satisfies property (a), the partition is IR for every~$v_{i'}$,~$i'\in[i]$. Most importantly, it must be the case that~$\val_{v_i}(\pttn^C_j) \geq 0$, which, by additivity, holds if and only if~$\sum_{w\in\pttn^C_j} \val_{v_i}(w) \geq 0$. We construct~$\pttn' = (\pttn^{i,\nVec,\uu}_j\setminus\{v_i\},\pttn^{i,\nVec,\uu})$ and claim that it is a valid extension of~$\pttn^C$ for~$(i-1,(s_{j}-1,\nVec_{-j}),(u_\ell - v_{w_\ell}(v_i)\cdot [| w_\ell \in \pttn_j^C |] )_{\ell\in[\vartheta]})$ and hence, by the induction hypothesis,~$\DP[i-1,(s_{j}-1,\nVec_{-j}),(u_\ell - v_{w_\ell}(v_i)\cdot [| w_\ell \in \pttn_j^C |] )_{\ell\in[\vartheta]}] = \texttt{true}$.
        If it is indeed the case, then~$\DP[i,\nVec,\uu] = \texttt{true}$ by the definition of the computation, as all properties of the recurrence are satisfied for~$j$.
    \end{claimproof}

    Once the dynamic programming table~$\DP$ is computed, we just check whether there exists a pair~$(\nVec,\uu)$ with~$s_j \geq 1$ for every~$j\in[\numCoals]$ ($\nVec = \sVec$ in the case of ASHGs-SCC) and~$u_i \geq 0$ for every~$i\in[\vartheta]$ such that~$\DP[n-\vartheta,\nVec,\uu] = \texttt{true}$. In other words, whether we can extend the initial partial partition~$\pttn^C$ with agents of~$I$ such that the resulting partition~$\pttn$ (a) is IR for all agents in~$I$ (by the definition of~$\pttn^{n-\vartheta,\nVec,\uu}$) and, by our choice of~$\uu$, also for all agents in~$C$ and (b) by our choice of~$\nVec$, all coalitions are of the correct size (non-empty in case of ASHGs and~$\sVec$ in case of ASHGs-SCC). Also, observe that, apart from the unary encoding, we have no assumption on the valuations, so our algorithm works even in the case of non-symmetric valuations.

    For the running time, there are~$\Oh{\numAgents}\cdot \numAgents^\numCoals \cdot (\Delta\cdot \val_{\max})^\Oh{\vartheta}$
    different cells of~$\DP$. Since~$\numCoals \leq \vartheta$,~$\Delta \in \Oh{\numAgents}$, and~$\val_{\max} \in n^\Oh{1}$, the number of cells can be upper-bounded by~$\numAgents^\Oh{\vartheta}$. As every cell can be computed in~$\Oh{\numAgents}$ time, the overall running time is~$\numAgents^\Oh{\vartheta}$, which is clearly in \XP.

    If~$\numCoals > \vartheta$, then we distinguish two cases. First, in the case of ASHGs, we observe that if~$\numCoals > \vartheta$, then~$\mathcal{I}$ is always a \Yes-instance. Specifically, for every~$v\in C$, we create a coalition~$\{v\}$. By this, we obtain~$\vartheta$ coalitions that are clearly individually rational. Then, we arbitrarily partition the vertices of~$I$ into~$\numCoals-\vartheta$ coalitions. Since~$I$ is an independent set of~$G$, the utility of every~$v\in I$ is zero, which is the same as being in the singleton coalition. Thus, all constructed coalitions are individually rational.

    It remains to show how to deal with the case of ASHGs-SCC. For this setting, we use the same dynamic programming as in the case of~$\numCoals \leq \vartheta$. Again, we guess a partial partitioning~$\pttn^C$ of~$C$. There are~$k^\Oh{\vartheta}\in n^\Oh{\vartheta}$ different such partitions~$\pttn_C$. Then, we adjust the dynamic programming from the case of~$\numCoals \leq \vartheta$ as follows. Instead of having a vector~$\nVec = (s_1,\ldots,s_\numCoals)$, we have a vector~$\nVec = (s_1,\ldots,s_{\numCoals'},s_{\numCoals'+1})$, where~$\numCoals'\leq \vartheta$ is the number of coalitions that are non-empty according to~$\pttn^C$. Let~$\operatorname{id}\colon [\numCoals]\to[\numCoals'+1]$ be a function such that for every~$\pttn^C_j$ that is nonempty, it returns an unique number from the interval~$[1,\numCoals']$, and for every~$\pttn^C_j$ that is empty, it returns~$\numCoals'+1$. The semantics of~$\nVec$ is that for every~$i\in[\numCoals']$,~$s_i$ is the number of agents in coalition~$\pttn_{\operatorname{id^{-1}(i)}}$ (i.e., in coalitions containing at least one vertex cover agent), and~$s_{\numCoals'+1}$ contains the number of agents that are in coalitions only with other agents of~$I$. The crucial observation here is that we can split the agents which are not in a coalition with vertex cover agents arbitrarily. Now, we run the same dynamic programming algorithm as in the case of~$\numCoals \leq \vartheta$. Finally, we ask whether there exist~$\nVec = (n_{\operatorname{id}^{-1}(1)},\ldots,n_{\operatorname{id}^{-1}(\numCoals')},\numAgents - \sum_{i=1}^{\numCoals'}n_{\operatorname{id}^{-1}(i)})$ and~$\uu = (u_1,\ldots,u_\vartheta)$ with~$u_i \geq 0$ for every~$i\in[\vartheta]$ such that~$\DP[\numAgents - \vartheta, \sVec, \uu] = \texttt{true}$. Again, there are~$\Oh{n} \cdot n^\Oh{\vartheta} \cdot (\Delta\cdot \val_{\max})^\Oh{\vartheta} \in \numAgents^\Oh{\vartheta}$ cells and each can be computed in polynomial time, which leads the promised running time. The correctness of this approach is then analogous to the correctness of the case with~$\numCoals \leq \vartheta$.
\end{proof}

The algorithm from the previous result shows that our problems are in \XP when parameterized by the vertex cover number of~$G$. In the following, we show that the algorithm is tight. Not only that, under standard theoretical assumption, no \FPT algorithm is possible, but also, the algorithm is asymptotically best possible.

\begin{theorem}\label{thm:ASHG:VC:unary:Whard}
    It is \Wh when parameterized by the number of coalitions~$\numCoals$, the number of negative edges, and the vertex cover number of~$G$, combined, to decide whether a~$\numCoals$-ASHG or ASHG-SCC~$\Gamma$ admits an individually rational coalition structure, even if the valuations are symmetric, encoded in unary, and~$G$ is a split graph. Moreover, unless ETH fails, there is no algorithm running in~$\numAgents^{o(\vc(G))}$ for these problems.
\end{theorem}
\begin{proof}
    We show the result by a parameterized reduction from the \probName{Balanced Bin Packing} problem. Here, we are given a multiset~$A=\{a_1,\ldots,a_\mu\}$ of integers, the number of bins~$B$, and the capacity~$C$ of every bin. The goal is to decide whether an assignment~$\alpha\colon A\to [B]$ exists so that for every bin~$j\in[B]$ we have~$\sum_{a\in\alpha^{-1}(j)} a \leq C$ and~$|\alpha^{-1}(j)| = \mu/B$. It is known that \probName{Balanced Bin Packing} is \Wh when parameterized by the number of bins~$B$, even if we assume that~$\sum_{a\in A} a  = B\cdot C$ and all integers are encoded in unary~\citep{KouteckyZ2025}. Without loss of generality, we can assume that~$\mu$ is divisible by~$B$, as otherwise, the instance of \probName{Balanced Bin Packing} is trivially a \No-instance.
    
    Given an instance~$\mathcal{I}=(A,B,c)$ of the \probName{Balanced Bin Packing} problem, we construct an equivalent instance~$\mathcal{J}=(\agents,(\val_i)_{i\in\agents},\numCoals)$ of the ASHGs as follows. First, we create~$2B$ \emph{bag agents}~$b^1_1,b^2_1\ldots,b^1_B,b^2_B$. For every~$j\in[B]$, we set~$\val_{b_j^1}(b_j^2) = \val_{b_j^2}(b_j^1) = -C$, and for every other pair, we set the value to~$-(B\cdot C + 1)$. Observe that the bag agents induce a complete graph. Next, for every item~$a_i\in A$, we create a corresponding \emph{item agent}~$v_i$. Every item agent~$v_i$ has a nonzero value only for bag agents. Specifically, we set~$\val_{v_i}(b_j^\ell) = \val_{b_j^\ell}(v_i) = a_i$ for every~$j\in[B]$ and~$\ell\in[2]$. That is, the item agents induce an independent set of~$G$. To finalize the construction, we set~$\numCoals = B$.

    For correctness, assume first that~$\mathcal{I}$ is a \Yes-instance and~$\alpha$ is a solution assignment. For every~$j\in[B]$, we create a coalition~$\pttn_j = \{b_j^1,b_j^2\} \cup \{ v_i \mid \alpha(a_i) = j \}$. Clearly, all coalitions are non-empty, and every agent belongs to exactly one coalition. It remains to show that such a coalition structure is individually rational. The item agents have only positive valuations, so for them, IR is clearly satisfied. Every bag agent~$b_j^\ell$ is in the same coalition with~$b_j^{\ell'}$, where~$\ell' = \{1,2\}\setminus\ell$. Assume that~$\val_{b_j^\ell}(\pttn) < 0$. Then the value this agent gets from the item agents allocated to the same coalition is smaller than~$C$, which contradicts that~$\alpha$ is a solution for~$\mathcal{I}$, since~$\sum_{v_i\in\pttn_j\setminus\{b_j^1,b_j^2\}} \val_{b_j^\ell}(v_i) = \sum_{v_i\in\pttn_j\setminus\{b_j^1,b_j^2\}} a_i = \sum_{a_i \mid \alpha(a_i) = j} a_i < C$. Consequently, the gain from item agents is at least~$C$, meaning that the value of~$v_j^\ell$ in~$\pttn$ is non-negative. That is,~$\pttn$ is an individually rational coalition structure.
    
    In the opposite direction, assume that~$\mathcal{J}$ is a \Yes-instance and~$\pttn$ is an individually rational coalition structure. First, we show that there is no pair of bag agents~$b_j^\ell$ and~$b_{j'}^{\ell'}$, where~$j\not=j'$ and~$\ell,\ell'\in\{1,2\}$, such that~$\pttn(b_j^\ell) = \pttn(b_{j'}^{\ell'})$. If this were the case, then the utility of agent~$b_j^\ell$ is at most 
   ~$\val_{b_j^\ell}(\pttn) = -(B C+1) + \sum_{v_i\in \pttn_j\setminus\{b_j^\ell,b_{j'}^{\ell'}\}} \val_{b_j^\ell}(v_i) 
    	\leq -B C - 1 + \sum_{v_i\in\agents\setminus\{b_1^1,\ldots,b_N^2\}} \val_{b_j^\ell}(v_i) 
    	= -B C - 1 + \sum_{v_i\in\agents\setminus\{b_1^1,\ldots,b_N^2\}} a_i
    	= -B C - 1 + \sum_{i=1}^N a_i
    	= -B C - 1 + B C = -1$.
    This contradicts that~$\pttn$ is individually rational. Consequently, by the Pigeonhole principle, as there are~$\numCoals$ coalitions and~$2\numCoals$ bag agents, it must be the case that for every~$j\in[B]$ we have~$\pttn(b_j^1) = \pttn(b_j^2)$. Also, without loss of generality, we can assume that~$\pttn(b_j^\ell) = j$, as we can always rename the coalitions. Now, we construct a solution~$\alpha$ for~$\mathcal{I}$. Specifically, for every~$a_i\in A$, we set~$\alpha(a_i) = \pttn(v_i)$. All items are clearly allocated, so it remains to show that, under the assignment~$\alpha$, the capacity of no bag is exceeded. For the sake of contradiction, assume that there is a bin~$j'\in[B]$ such that~$\sum_{a\in \alpha^{-1}(j')} a > C$. Then, since~$\sum_{a\in A} a = B\cdot C$, there is at least one bin, say~$j$, with~$\sum_{a\in\alpha^{-1}(j)} a < C$. However, if such a bin exists, then the utility of~$b_{j}^1$ in~$\pttn$ is~$\val_{b_j^1}(\pttn) = \sum_{v\in \pttn^{-1}(j)\setminus\{b_j^1\}} \val_{b_j^1}(v)
    	= \val_{b_j^1}(b_j^2) + \sum_{v_i\in\pttn^{-1}(j)\setminus\{b_j^1,b_j^2\}} \val_{b_j^1}(v_i)
    	= -C + \sum_{v_i\in\pttn^{-1}(j)\setminus\{b_j^1,b_j^2\}} a_i
    	= -C + \sum_{a_i\in\alpha^{-1}(j)} a_i < 0$, 
    which contradicts that~$\pttn$ is individually rational. Therefore, neither such~$j$ nor~$j'$ can exist, and we obtain that~$\alpha$ is indeed a solution for~$\mathcal{I}$.

    To wrap up, observe that we used~$\numCoals = B$, there are~$2B$ bag agents who form a vertex cover of~$G$, and all negative edges are inside the clique on bag agents---that is, there are~$\binom{B}{2}$ negative edges. %
    For ASHGs-SCC, we create an instance with an identical graph, identical valuations, and identical~$\numCoals$, and set~$\lVec(j)=\uVec(j)=\mu/B$ for every~$j\in[\numCoals]$.

    For the ETH-based lower bound, it was recently shown that \probName{Balanced Bin Packing} does not admit an algorithm running in $\mu^{o(B)}$ time unless ETH fails~\citep{BringmannDW2026}. For the sake of contradiction, assume that there is an algorithm~$\mathbb{A}$ running in $\numAgents^{o(\vc(G))}$ time deciding whether an individually rational solution exists. Then, given an instance $\mathcal{I}$ of \probName{Balanced Bin Packing}, we use the reduction described above to obtain an equivalent instance~$\mathcal{J}$, use $\mathbb{A}$ to decide $\mathcal{J}$, and return the same response for $\mathcal{I}$. The response is correct by the correctness of the reduction. Moreover, as we can create $\mathcal{J}$ in polynomial time, we hav an algorithm for \probName{Balanced Bin Packing} running in $n^\Oh{1} + \numAgents^{o(\vc(G))} = \numAgents^{o(2B)}$, which contradicts ETH.
\end{proof}

Notice that the hardness result is, in fact, even stronger, and shows \Whness with respect to the combined parameter of the number of coalitions~$\numCoals$, the number of negative edges, and the vertex cover number of~$G$.

It turns out that, for fixed-parameter tractability, we have to restrict the valuations even more. In particular, there is an \FPT algorithm for the parameterization by the vertex cover number of~$G$ and maximum value in the valuation functions. The algorithm utilize N-fold ILP as its sub-procedure.

\begin{theorem}\label{thm:ASHG:FPT:VC:Wmax}
    It can be decided in \FPT time with respect to the vertex cover number of~$G$ and~$\wFn_{\max}$ whether an instance of~$\numCoals$-ASHGs admits an individually rational coalition structure.
\end{theorem}
\begin{proof}
    In our algorithm, we can assume that a vertex cover~$C$ of the optimal size~$\vartheta=\vc(G)$ is given as part of the input, as, if not, then we can compute one in~$1.25284^{\vartheta}\cdot n^\Oh{1}$ time~\citep{HarrisN2024}. Recall that, by the argumentation of \Cref{thm:ASHG:VC:unary:XP}, we can assume that~$\numCoals \leq \vartheta$, as otherwise, an IR coalition structure is guaranteed to exist.

    Our algorithm is based on a guess of a partitioning of the vertex cover agents combined with an ILP formulation that verifies whether the guess can be extended with agents outside~$C$ in an individually rational way. Specifically, we model the extension sub-procedure as an N-fold ILP.

    More formally, let~$\pttn^C$ be one of~$\numCoals^\Oh{\vartheta} \in 2^\Oh{\vartheta\log\vartheta}$ possible partitionings of the vertex cover agents into~$\numCoals$ coalitions. For every vertex~$v\in V(G)\setminus C$, we compute the set~$I_v = \{ i \in[\numCoals] \mid \sum_{u\in\pttn^C_i} \wFn(v,u) \geq 0 \}$, i.e.,~$I_v$ is the set of coalitions such that if we add~$v$ to~$\pttn_i^C$, it is individually rational for~$v$. Observe that this can be decided only on the basis of the vertex cover vertices in~$\pttn^C_i$, as~$V(G)\setminus C$ is an independent set. 
    
    In our N-fold ILP, we have a binary variable~$x_{v,i}$ for every~$v\in V(G)\setminus C$ and every coalition~$i\in I_v$. The constraints of the ILP are as follows. First, we have a single local constraint for every~$v\in V(G)\setminus C$, which ensures that each agent outside the vertex cover is assigned to exactly one coalition. Formally, the local constraint is
    \begin{align}
        \forall v\in V(G)\setminus C && \sum_{i\in I_v} x_{v,i} = 1\,. &&\label{eq:ASHG:vc:wmax:FPT:agentCoal}
    \end{align}
    Next, we add a set of global constraints that ensure that (a) the coalition structure is individually rational for every vertex cover agent~$u\in C$ and (b) none of the~$\numCoals$ coalitions is empty. The first condition can be encoded as
    \begin{align}
        \forall u\in C && \val_u(\pttn^C) + \sum_{\mathclap{\substack{v \in V(G)\setminus C\\\pttn^C(u)\in I_v}}} x_{v,\pttn^C(u)}\cdot \wFn(u,v) \geq 0\,, &&\label{eq:ASHG:vc:wmax:FPT:coverVerticesIR}
    \end{align}
    while the second condition can be enforced by introducing the following set of constraints
    \begin{align}
        \forall i\in[\numCoals] && |\pttn^C_i| + \sum_{\mathclap{\substack{v\in V(G)\setminus C\\i\in I_v}}} x_{v,i} > 0\,.\label{eq:ASHG:vc:wmax:FPT:coalSizeAtLeastOne}
    \end{align}
    For the running time, observe that each variable participates in exactly one local constraint, there are~$\vartheta + \numCoals \in \Oh{\vartheta}$ global constraints, and the maximum coefficient of a variable is~$\wFn_{\max}$. Therefore, the N-fold ILP can be solved in time~$(\wFn_{\max})^\Oh{\vartheta^2}\cdot \numAgents^\Oh{1}$, which is clearly in \FPT. As we need to run the ILP for every possible~$\pttn^C$, the overall running time of the algorithm is~$2^\Oh{\vartheta\log\vartheta}\cdot (\wFn_{\max})^\Oh{\vartheta^2}\cdot \numAgents^\Oh{1} \in 2^{\vartheta^2\cdot\log(\wFn_{\max})} \cdot \numAgents^\Oh{1}$. Also, observe that we do not have any requirement on the symmetry of valuations.

    For correctness, assume first that there is an individually rational coalition structure~$\pttn$. Let~$\pttn^C = (\pttn^C_1,\ldots,\pttn^C_\numCoals)$, where~$\pttn^C_i = \pttn_i \cap C$. For every~$v\in V(G)\setminus C$, we set~$x_{v,i} = 1$ if and only if~$i = \pttn(v)$, and we claim that it is a feasible solution for the N-fold ILP constructed for~$\pttn^C$. As~$\pttn$ is individually rational, it holds that for every~$v\in V(G)\setminus C$ we have~$\val_v(\pttn) = \val_v(\pttn(v)) = \sum_{u\in \pttn_{\pttn(v)}} \wFn(v,u) = \sum_{u \in \pttn_{\pttn(v)} \cap C} \wFn(v,u) = \sum_{u\in \pttn^C_{\pttn(v)}} \wFn(v,u) \geq 0$, meaning that~$\pttn(v)\in I_v$. Hence, for each vertex~$v\in V(G)\setminus C$, we have exactly one~$x_{v,i}$ equal to one. Consequently, all constraints of type \eqref{eq:ASHG:vc:wmax:FPT:agentCoal} are satisfied. Moreover, as~$\pttn$ is a solution, it must be the case that each~$\pttn_i$ is of size at least one, and hence, also constraints of type \eqref{eq:ASHG:vc:wmax:FPT:coalSizeAtLeastOne} are satisfied. It remains to show that the constraints \eqref{eq:ASHG:vc:wmax:FPT:coverVerticesIR} are satisfied. By the construction of~$x_{v,i}$, we have that for every~$u\in C$, the constraint \eqref{eq:ASHG:vc:wmax:FPT:coverVerticesIR} can be rewritten as~$\val_u(\pttn^C) + \sum_{v\in \pttn^C_{\pttn(u)}\setminus C} \wFn(u,v) = \sum_{v \in \pttn_{\pttn(u)}} \wFn(u,v)$. However, since~$\pttn$ is IR, this must be at least zero, and hence, \eqref{eq:ASHG:vc:wmax:FPT:coverVerticesIR} is satisfied for every~$u\in C$.

    In the opposite direction, let~$\pttn_C$ be a partition of the vertex cover agents such that there is a feasible solution~$\vec{x}$ of the corresponding ILP. We construct~$\pttn = (\pttn_1,\ldots,\pttn_\numCoals)$ so that we set~$\pttn_i = \pttn^C_i \cup \{ v\in V(G)\setminus C \mid x_{v,i} = 1 \}$. By \eqref{eq:ASHG:vc:wmax:FPT:agentCoal}, for each vertex~$v$, there is exactly one~$x_{v,i}$ equal to one, and by constraint \eqref{eq:ASHG:vc:wmax:FPT:coalSizeAtLeastOne}, each coalition is of size at least one. That is, the partition~$\pttn$ is well-defined. Moreover, since~$x_{v,i} = 1$ if and only if~$i\in I_v$, each agent~$v\in V(G)\setminus C$ has utility at least zero in~$\pttn$. It remains to be shown that the coalition structure is IR also for all vertex cover agents. For the sake of contradiction, let~$u\in C$ be a vertex cover agent such that~$\val_u{\pttn} < 0$. We can decompose the utility as~$\val_u(\pttn) = \val_u(\pttn(u)) = \val_u(\pttn^C(u)) + \sum_{v\in\pttn(u)\setminus C} \wFn(u,v) = \val_u(\pttn^C(u)) + \sum_{v\in\pttn(u)\setminus C} x_{v,\pttn(u)}\cdot \wFn(u,v)$, which is exactly equal to the right side of constraint \eqref{eq:ASHG:vc:wmax:FPT:coalSizeAtLeastOne}. This contradicts that~$\vec{x}$ is a feasible solution. Therefore,~$\pttn$ must be individually rational also for all vertex cover agents, finishing the correctness of the algorithm.
\end{proof}

Notice that the previous algorithm work only for the~$\numCoals$-ASHGs. The obstacle for generalizing it for the ASHGs-SCC is that we cannot easily resolve the case when the number of coalitions is greater than the vertex cover number of~$G$. Therefore, we obtain the following, weaker result for this variant of the problem.

\begin{proposition}\label{thm:SCC:FPT:VC:k:Wmax}
    When parameterized by the vertex cover number~$\vc(G)$, the number of coalitions~$\numCoals$, and the maximum value~$\wFn_{\max}$, the problem of deciding whether a ASHG-SCC~$\Gamma$ admits an individually rational coalition structure is fixed-parameter tractable.
\end{proposition}
\begin{proof}
    If we are in the setting of ASHGs-SCC, that is, we have prescribed upper and lower bound of each coalition, and~$k\leq \vartheta$, we just replace the constraints \eqref{eq:ASHG:vc:wmax:FPT:coalSizeAtLeastOne} with the following:
    \begin{align}
        \forall j\in[\numCoals] && \lVec(j) \leq |\pttn^C_j| + \sum_{\mathclap{\substack{v\in V(G)\setminus C\\j\in I_v}}} x_{v,j} \leq \uVec(j)\,.\label{eq:ASHG:vc:wmax:FPT:coalSizeExactly}
    \end{align}
    The time complexity and the correctness of the algorithm then remain the same.
\end{proof}

\section{Tree-Like Preference Graph}

In this section, we focus on graphs that are sparse (from a graph-theoretical perspective). This is a reasonable approach in view of Theorem~\ref{thm:ASHG:NPc:clique}: both variants of our problem become intractable already if the underlying graph~$G$ is a complete graph---arguably the ``easiest'' structure between dense graphs.

We explore two different dimensions. First, the following hardness proves that the algorithm of \Cref{thm:ASHG:FPT:VC:Wmax} cannot be strengthened to more general graph families, as both variants of ASHGs are \Wh when parameterized by the treedepth of~$G$, even if the valuations are binary. Moreover, the reduction asks for partitioning into two coalitions. This strengthen hardness from \Cref{thm:ASHG:VC:NPhard} and shows that even if the number of coalitions is constant, polynomial-time algorithm cannot exist regardless of how restricted the preferences are. Additionally, the result is tight, as if~$\numCoals = 1$, the only possible solution is the grand coalition, and we can check whether it is IR in polynomial time. The reduction is from the \probName{General Factors} problem~\citep{GutinKSSY2012}, where we are given a bipartite graph~$H=(X,Y)$, along with a list function~$L\colon V(H) \to 2^{[\Delta(H)]}$, and the goal is to decide if there exists a subset~$S\subseteq E(H)$ such that~$d_{H-S}(u) \in L(u)$ for all~$u \in V(H)$.

\begin{theorem}\label{thm:ASHG:TD:binary:Whard}
    When parameterized by the treedepth of~$G$, it is \Wh to decide if a~$\numCoals$-ASHG or ASHG-SCC~$\Gamma$ admits an individually rational coalition structure, even if the valuations are symmetric, binary, and~$\numCoals = 2$.
\end{theorem}
\begin{proof}
    We show this result by providing a parameterized reduction from \probName{General Factors} problem on bipartite graphs, where it is known to be \Wh parameterized by the size of the smallest partition~\citep{GutinKSSY2012}. In an instance of \probName{General Factors} (on bipartite graphs), input is a bipartite graph~$H=(X\cup Y,E)$, along with a list function~$L\colon V(H) \to \mathcal{P}(\{0,\dots,\Delta(H)\})$, and the goal is to decide if there exists a subset~$S\subseteq E(H)$ such that~$d_{H-S}(u) \in L(u)$ for all~$u \in V(H)$.  

    From an instance~$(H=(X\cup Y,E), L)$, we construct an equivalent instance~$\mathcal{J}=(\agents,(\val_i)_{i\in\agents},2)$ of the ASHGs such that for distinct~$i,j\in \agents$,~$\val_i(j) \in \{+1,-1,0\}$ as follows. Moreover, our valuations will be symmetric and hence, the valuations for our instance~$\mathcal{J}$ are represented using an undirected graph~$G = (\agents,E,\wFn)$, where~$ij\in E$ iff~$\wFn(ij) \in \{+1,-1\}$. Our reduction requires several gadgets, which we will provide one by one to ease the exposition. Moreover, we are targeting a coalition structure with two coalitions, and we refer to these coalitions as~$\pi_1$ and~$\pi_2$.
    \begin{figure}
        \centering
        \begin{tikzpicture}[
            vx/.style={draw, circle, inner sep=0pt, minimum width=16pt, thick, fill=white},
            redge/.style={very thick,cbRed},
            gedge/.style={very thick, densely dashed,cbGreen},
          ]
         
          \node[vx] (x8) at (-1.5, 1.5)    {$x_8$};
          \node[vx] (x9) at (-3, 0)   {$x_9$};
          \node[vx] (x1) at (0, 0) {$x_1$};
          \draw[redge] (x8) -- (x9);
          \draw[redge] (x8) -- (x1);
          \draw[gedge] (x9) -- (x1);
         
          \node[vx] (x5) at (3, 0)   {$x_5$};
          \node[vx] (x2) at (1.5, -1)   {$x_2$};
          \node[vx] (x6) at (3, -2)   {$x_6$};
          \node[vx] (x3) at (1.5, -3)  {$x_3$};
          \node[vx] (x7) at (3, -4)  {$x_7$};
          \node[vx] (x4) at (1.5, -5)  {$x_4$};
          \draw[redge] (x1) -- (x5);
          \draw[redge] (x5) -- (x2);
          \draw[redge] (x2) -- (x6);
          \draw[redge] (x6) -- (x3);
          \draw[redge] (x3) -- (x7);
          \draw[redge] (x7) -- (x4);
         
          \node[vx] (x10) at (-0.75, -5)  {$x_{10}$};
          \node[vx] (x12) at (-3, -5)   {$x_{12}$};
          \node[vx] (x11) at (0.25, -6.5) {$x_{11}$};
          \draw[redge] (x4) -- (x10);
          \draw[redge] (x4) -- (x11);
          \draw[redge] (x10) -- (x11);
          \draw[gedge] (x12) -- (x10);
         
        \end{tikzpicture}
        \caption{Illustration for vertex gadget~$V^x$. Here, each green (dashed) edge corresponds to a weight of~$+1$ and each red (solid) edge corresponds to a weight of~$-1$.}
        \label{fig:vGadget}
    \end{figure}
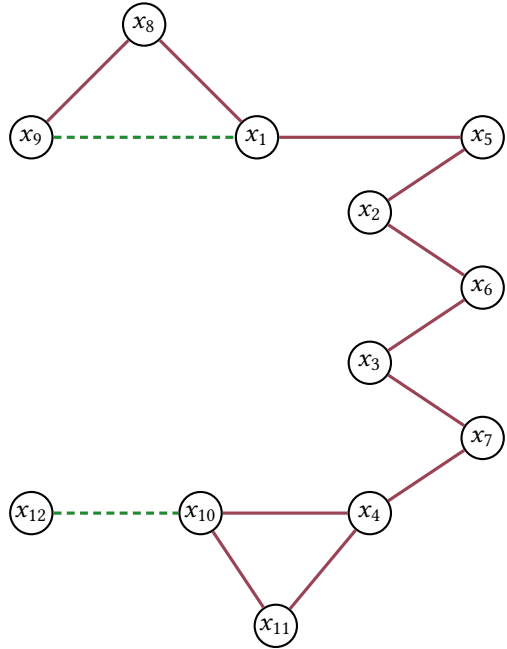
    
    \smallskip
    \noindent\textbf{Vertex Gadget.} First, we have a \textit{vertex gadget}, that will be used to replace vertices of~$H$. In particular, for each~$x\in V(H)$, we have a vertex gadget~$V^x$ that consists of~$12$ vertices~$\{x_1,x_2,\ldots,x_{12}\}$. See \Cref{fig:vGadget} for an illustration. Next, we define edges/valuations of our vertex gadget~$V^x$: for~$j\in [5,7]$,~$\wFn(x_jx_{j-4}) = \wFn(x_jx_{j-3}) = -1$;~$\wFn(x_8x_1)= \wFn{x_8}(x_8x_9)= \wFn(x_{11}x_4)= \wFn(x_{11}x_{10})= \wFn(x_{10}x_4)=-1$; and~$\wFn(x_9x_1)= \wFn(x_{12}x_{10}) = +1$. We say that~$x_{12}$ is the \textit{leaf vertex} of the vertex gadget~$V^x$ While, extending our graph, only vertices~$x_1,x_2,x_3,$ and~$x_4$ will have edges outside the gadget~$V^x$, and hence, we can deduce some properties of any IR stable coalition right now. We have the following claim:
    \begin{claim}\label{C:basic}
        In any IR stable solution, vertices~$x_1,x_2,$~$x_3,x_4,x_9,x_{10},x_{12}$ must be in the same coalition. 
    \end{claim}
    \begin{claimproof}
        Recall that any solution will have~$2$ coalitions, say~$\pi_1,\pi_2$ and only vertices~$x_1,\ldots,x_4$ can have edges with the rest of the graph. The proof of our claim follows from the observation that if an agent only have~$-1$ valuation for all of its neighbors, then all of its neighbors will be in the coalition distinct from its own.   Without loss of generality, let us assume that~$x_1\in \pi_1$. Thus,~$x_5 \in \pi_2$, implying~$x_2\in \pi_1$, implying~$x_6\in \pi_2$, implying~$x_3 \in \pi_1$, implying~$x_7 \in \pi_2$, implying~$x_4\in \pi_1$. Thus,~$x_1,x_2,x_3,x_4\in \pi_1$. Similarly,~$x_1\in \pi_1$ implies~$x_8\in \pi_2$, implying~$x_9 \in \pi_1$, and~$x_4\in \pi_1$ implies~$x_{11}\in \pi_2$, implying~$x_{10} \in \pi_1$. Finally, since~$\wFn (x_{10}x_4) = -1$, there must be at least one neighbor of~$x_{10}$ with positive valuation in~$\pi_1$, and since~$x_{12}$ is the unique such vertex, we have that~$x_{12} \in \pi_1$. This completes the proof of our claim. We remark here that~$\pi_1$ is not individually stable for~$x_4$ yet, but  we will add more edges incident on~$x_4$ later that will make the coalition~$\pi_1$ stable or unstable for~$x_4$ depending on whether~$(H,L)$ is a Yes or No instance of \probName{General Factors}.  
    \end{claimproof}

    Next, we add an ``enforcer'' vertex~$s$ to~$G$ and for each~$u\in V(H)$, we connect~$s$ to~$u_4$ such that~$\wFn(su_4) = -1$.  This immediately leads us to the following observation: In any IR coalition structure,~$s$ is in one of the coalitions and, for each~$u\in V(H)$,~$u_4$ is in the other coalition. This observation, along with \Cref{C:basic} leads, us to make the following remark.

    \begin{remark}\label{R:1}
        Without loss of generality, let us assume for the rest of proof that~$s\in \pi_2$. Thus, for each vertex~$x \in V(H)$,~$x_4\in \pi_1$, and hence,~$x_1,x_2,$~$x_3,x_4,x_9,x_{10},x_{12} \in \pi_1$ and~$x_5,x_6,x_7,x_8,x_{11} \in \pi_2$. Thus, when  restricted to vertex gadgets and the enforcer vertex, in any IR coalition structure,~$|\pi_1| = 7 |V(H)|$ and~$|\pi_2| = 5 |V(H)|+1$. 
    \end{remark}

    \noindent\textbf{List Element Gadget.} For a vertex~$u\in V(H)$, let~$\ell_u = \max L(u)$. For each~$u\in V(H)$ and for each element~$\ell \in L(u)$, we create a \textit{list element gadget}~$L^{\ell,u}$ in the following manner. Consider a star on~$\ell_u$ leaves and subdivide each edge of this star exactly once. Next, corresponding to each edge~$e$ of this subdivided star, we have~$\wFn(e) = -1$. We will refer to the center of the star as the \textit{center vertex of}~$L^{\ell,u}$. Observe that in any stable coalition, the leaves and the center of~$L^{\ell,u}$ (in total,~$\ell_u +1$ agents) will end up in the same coalition, and all neighbors of leaves ($\ell_u$ vertices introduced by the subdivision) will end up in the other coalition.

    \smallskip
    \noindent\textbf{Edge Gadget.} Finally, for every edge~$e\in E(H)$, we will have an \textit{edge gadget}~$E^e$, which is a cycle on four vertices, say~$w_1,w_2,w_3,w_4$, such that~$\wFn(w_1w_2) = \wFn(w_2w_3) = \wFn(w_3w_4)= \wFn(w_4w_1) = -1$. We refer to~$w_2,w_4$ as the \textit{center vertices of}~$E^e$ and~$w_1,w_3$ as the endpoints of~$E^e$. Observe that in any (IR) stable coalition structure, for any edge~$e$, the endpoints of the edge gadget~$E^e$ (i.e.,~$w_1,w_3$) will always end up in the same coalition while the center vertices of the edge gadget (i.e.,~$w_2,w_4$) will end up in the other coalition. Thus, each edge gadget will contribute exactly two vertices to~$\pi_1$ and exactly two vertices to~$\pi_2$. Thus, we can extend our \Cref{R:1} to the following. 
    \begin{remark}\label{R:2}
        When restricted to the vertex gadgets, edge gadgets, and the enforcer vertex~$s$,~$|\pi_1| = 7 |V(H)| + 2|E(H)|$ and~$|\pi_2| = 5|V(H)| + 2|E(H)|+1$. 
    \end{remark}

    Next, we  provide full construction of our ASHG-SCC instance using~$(H,L)$ and the gadgets defined above. We begin with replacing each vertex~$x\in V(H)$ with a copy of the vertex gadget~$V^x$ having 12 vertices~$x_1,\ldots,x_{12}$. Next, for each vertex~$x\in V(H)$ and list element~$\ell \in L(x)$, we construct a list element gadget~$L^{\ell,x}$. At this point, for each~$x\in V(H)$, we connect the central vertex, say~$z$, of each~$L^{\ell,x}$ to~$x_1$ and~$x_4$ such that~$\wFn(zx_1) = -1$ and~$\wFn(zx_4) = +1$. Moreover, for arbitrary~$\ell$ leaf vertices (out of~$\ell_x$ leaf vertices of~$L^{\ell,x}$)~$y_1,\ldots,y_{\ell}$,  we connect~$y_i$ (for~$i\in [\ell]$) to~$x_2$ and~$x_3$ such that~$\wFn(y_ix_2) = +1$ and~$\wFn(y_ix_3) = -1$. The vertices~$x_1, x_4$ do not have any additional edges. (We will connect more edges to vertices~$x_2$ and~$x_3$). See \Cref{fig:complete} for an illustration. Before proceeding with the edge gadgets, we prove the following easy claim.
    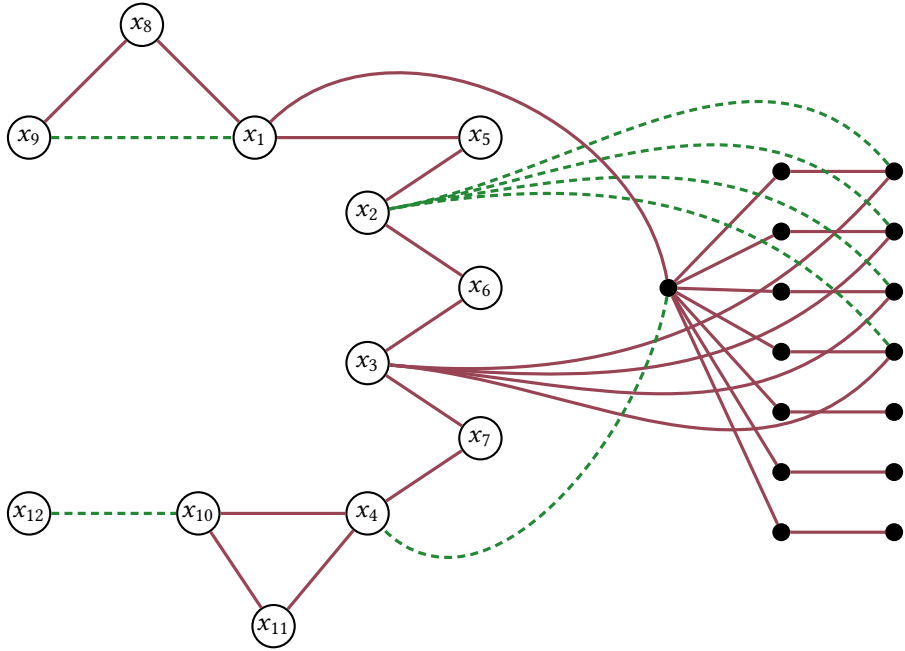
\begin{figure}
        \centering
        \begin{tikzpicture}[
            vx/.style={draw, circle, inner sep=0pt, minimum width=16pt, thick, fill=white},
            vl/.style={draw, circle, inner sep=0pt, minimum width=6pt, thick, fill=black},
            redge/.style={very thick,cbRed},
            gedge/.style={very thick, densely dashed,cbGreen},
            ]
            
            \node[vx] (x8) at (-1.5, 1.5)    {$x_8$};
            \node[vx] (x9) at (-3, 0)   {$x_9$};
            \node[vx] (x1) at (0, 0) {$x_1$};
            \draw[redge] (x8) -- (x9);
            \draw[redge] (x8) -- (x1);
            \draw[gedge] (x9) -- (x1);
            
            \node[vx] (x5) at (3, 0)   {$x_5$};
            \node[vx] (x2) at (1.5, -1)   {$x_2$};
            \node[vx] (x6) at (3, -2)   {$x_6$};
            \node[vx] (x3) at (1.5, -3)  {$x_3$};
            \node[vx] (x7) at (3, -4)  {$x_7$};
            \node[vx] (x4) at (1.5, -5)  {$x_4$};
            \draw[redge] (x1) -- (x5);
            \draw[redge] (x5) -- (x2);
            \draw[redge] (x2) -- (x6);
            \draw[redge] (x6) -- (x3);
            \draw[redge] (x3) -- (x7);
            \draw[redge] (x7) -- (x4);
            
            \node[vx] (x10) at (-0.75, -5)  {$x_{10}$};
            \node[vx] (x12) at (-3, -5)   {$x_{12}$};
            \node[vx] (x11) at (0.25, -6.5) {$x_{11}$};
            \draw[redge] (x4) -- (x10);
            \draw[redge] (x4) -- (x11);
            \draw[redge] (x10) -- (x11);
            \draw[gedge] (x12) -- (x10);
            
            \node[vl] (l0) at (5.5,-2) {};
            
            \foreach[count=\i] \y in {0.45,1.25,2.05,...,5.5} {
                \node[vl] (l\i) at (7,-\y) {};
                \node[vl] (l\i2) at (8.5,-\y) {};
                \draw[redge] (l0) -- (l\i) -- (l\i2);
            }
            
            \foreach \i in {1,2,3,4} {
                \draw[gedge,out=130,in=10] (l\i2) to (x2);
                \draw[redge,out=-131,in=-5] (l\i2) to (x3);
            }
            
            \draw[redge,out=100,in=45] (l0) to (x1);
            \draw[gedge,out=-100,in=-45] (l0) to (x4);
        \end{tikzpicture}
        \caption{Each red edge corresponds to a weight of~$-1$ and each green edge corresponds to a weight of~$+1$. Here, the list gadget has~$7$ leaves (i.e.,~$\ell_x =7$) but only (top)~$4$ of the leaves have edges to~$x_2,x_3$ (i.e.,~$\ell_i = 4$).}
        \label{fig:complete}
    \end{figure}
    
    \begin{claim}\label{C:vertexList}
        Consider a vertex~$x\in V(H)$, the list~$L(x)= \{\ell_1,\ldots,\ell_p\}$, and their corresponding gadgets~$V^x$ and~$L^{\ell_1,x},\ldots,L^{\ell_p,x}$, respectively. In any (IR) stable coalition structure, there is a unique~$i\in [p]$ such that the center vertex of~$L^{\ell_i,x}$, and hence all the leaves of~$L^{\ell_i,x}$, will be in the same coalition as the coalition containing~$x_1,x_2,x_3,x_4$. Moreover, if~$L^{\ell_i,x}$ is the unique list gadget whose center vertex and leaves are in~$\pi_1$, then the weight on~$x_2$ and~$x_3$ due to vertices of list gadgets is~$+\ell_i$ and~$-\ell_i$, respectively.
    \end{claim}
    \begin{claimproof}
        Due to \Cref{R:1}, we can assume without loss of generality that~$x_1,x_2,x_3,x_4,x_9,x_{10},x_{12} \in \pi_1$. Since in this (partial) coalition, the weight on~$x_1$ is~$+1$, there can be at most one vertex~$y$ in~$\pi_1$ such that~$\wFn(yx_1) = -1$. Since for each center vertex~$z$ of the list gadgets connected to~$x_1$,~$\wFn(zx_1) = -1$, at most one center vertex of the list gadgets connected to~$x_1$ can be in coalition~$\pi_1$. Similarly, observe that in the coalition restricted to ~$x_1,x_2,x_3,x_4,x_9,x_{10},x_{12}$, the weight on~$x_4$ is exactly~$-1$, there has to be at least one vertex~$y\in \pi_1$ such that~$\wFn(yx_4) = +1$. Since for each center vertex~$z$ of the list gadgets connected to~$x_4$,~$\wFn(zx_1) = +1$ (and~$x_4$ is not connected to any other vertices in~$G$), at least one center vertex of the list gadgets connected to~$x_1$ must be in coalition~$\pi_1$. Thus, there is exactly one~$i\in [p]$ such that the central vertex of~$L^{\ell_i,x}$ is in~$\pi_1$. Finally, since exactly~$\ell_i$ of~$\ell_u$ leaves are connected to~$x_2$ and~$x_3$ with edges with weights~$+1$ and~$-1$, respectively, we get the desired weights on~$x_2$ and~$x_3$.
    \end{claimproof}

    Due to \Cref{C:vertexList}, we can extend our \Cref{R:2} to the following remark, which establishes that in any IR coalition structure, the coalitions~$\pi_1$ and~$\pi_2$ will have a fixed size.
    \begin{remark}\label{R:3}
        In any IR coalition structure,~$|\pi_1| = 7|V(H)| +2 |E(H)| + \sum_{u\in V(H)} (|L(u)|\cdot \ell_u+1)$ and~$|\pi_2| = 5|V(H)| +2|E(H)|+1+\sum_{u\in V(H)} (|L(u)|\cdot \ell_u+|L(u)|-1)$.
    \end{remark}
    Due to \Cref{R:3}, we have that the sizes of the coalitions in any IR coalition structure are fixed, and hence to complete our proof, we only need to establish that the instance that we created is equivalent to the input instance of \textsc{General Factors}.
    
    Next, for each edge~$e = xy\in E(H)$, recall that we have an edge gadget~$E^e$ with vertices~$w_1,w_2,w_3,w_4$. We connect~$w_1$ to~$x_2,x_3$ such that~$\wFn(w_1x_2) = -1$ and~$\wFn(w_1x_3) = +1$. Similarly, we connect~$w_3$ to~$y_2$ and~$y_3$ such that~$\wFn(w_3y_2) = -1$ and~$\wFn(w_3y_3) = +1$. 
    
    Recall our coalition structure with coalitions~$\pi_1,\pi_2$. Observe that ~$\pi_2$ induces an independent set in~$G$ and hence is IR. We will in fact maintain that the final coalition~$\pi_2$ still induces an independent set in~$G$, and hence always respects IR.
    
    Next, consider the weights on vertices of~$\pi_1$ in this partial assignment. For each vertex that is a center vertex (resp. a leaf vertex) of a list gadget, it has a positive edge to~$x_4$ (resp.~$x_2$) and a negative edge to~$x_1$ (resp.~$x_3$) and no other edges, and hence respects IR. Similarly, all vertices coming from a vertex gadget~$V^x$ except~$x_2$ and~$x_3$ have non-negative weights as argued before and the weight on~$x_2$ is~$+\ell_i$ and weight on~$x_3$ is~$-\ell_i$ where ~$L^{\ell_i,x}$ is the list gadget whose center and leaves are in~$\pi_1$. At this point, the only vertices of~$G$ that are not assigned to a coalition are the vertices of the edge gadgets. Recall that for each edge gadget~$E^e$, either the center vertices of~$E^e$ will be assigned to~$\pi_1$ and the endpoints will be assigned to~$\pi_2$, or vice versa (i.e., the endpoints are assigned to~$\pi_1$ and the center assigned to~$\pi_2$).  Since all of these vertices (from edge gadgets) are independent to vertices of~$\pi_2$, for any edge gadget, we can either add its central vertex or its endpoints to~$\pi_2$ while maintaining that~$\pi_2$ respects individual rationality.  
    
    Now, to complete our proof, we need to establish that~$G$ admits an IR coalition structure if and only if~$(H,L)$ is a Yes-instance of \textsc{General Factors}. Since till this point, the partial allocations~$\pi_1$ and~$\pi_2$ are created without loss of generality and~$\pi_2$ will admit IR,~$G$ admits an IR coalition structure if and only if~$\pi_1$ can be extended to an IR coalition where for each edge gadget, we either add its center vertices to~$\pi_1$ or its endpoints to~$\pi_1$. 

     In one direction, let~$(H,L)$ is a Yes-instance of \textsc{General Factors}. Let~$S\subseteq E(H)$ be  such that~$d_{H-S}(u) \in L(u)$ for all~$u\in V(H)$. Now, for each~$x\in V(H)$, we choose~$L^{\ell_i,x}$ where~$\ell_i = d_{H-S}(u) \in L(u)$ as the unique list gadget such that the leaves and the center of~$L^{\ell_i,x}$ are in~$\pi_1$. At this point, the weight on~$x_2$ is~$+\ell_i$ and the weight on~$x_3$ is~$-\ell_i$ and the weight on all other vertices is positive (as discussed above). Finally, for each edge~$e\notin S$, assign the endpoints of the edge gadget in~$E^e$ to~$\pi_1$ and for each edge~$e'\in S$, assign the center vertices of~$E^{e'}$ to~$\pi_1$. Recall that the center vertices of edge gadgets are independent to~$\pi_1$ and exactly one of the endpoints of an edge gadget has weight exactly~$-1$ to~$x_2$ and weight exactly~$+1$ to~$x_3$ if and only if the corresponding edge is incident to~$x$. Thus, for each vertex~$x\in V(H)$, the vertices~$x_2$ and~$x_3$ receive a valuation of weight exactly~$-d_{H-S}(x)$ and~$+ d_{H-S}(u)$,  respectively. Thus, for~$x_2$ as well as~$x_3$, the total weight on them in~$\pi_1$ is~$0$. Finally, since the endpoints of each edge gadget if are connected to some~$x_2$ via an edge with weight~$-1$, then they are also connected to~$x_3$ with weight~$+1$, and since~$x_2$ and~$x_3$ are in the same coalition, each endpoint of an edge gadget will have a~$0$ valuation in its respective coalition. Thus~$G$ admits an IR coalition structure.

     In the other direction, let~$G$ admits an IR coalition structure with coalition~$\pi_1$ and~$\pi_2$. Recall that for each vertex~$x\in V(H)$,~$x_2,x_3\in \pi_1$. Further there is a unique list gadget~$L^{\ell_i,x}$ such that the center vertex and the leaves of~$L^{\ell_i,x}$ are in~$\pi_1$. Thus, the total weight on~$x_2$ and~$x_3$ till this point is~$+\ell_i$ and~$-\ell_i$, respectively. At this point, we claim that there are exactly~$\ell_i$ edge gadgets corresponding to edges incident on~$x$ whose endpoints are in~$\pi_1$. To see this, observe that if more than~$\ell_i$ edge gadgets corresponding to edges incident on~$x$ have their endpoints in~$\pi_1$, then~$x_2$ will have a negative weight, and if less than~$\ell_i$ edge gadgets corresponding to edges incident on~$x$ have their endpoints in~$\pi_1$, then~$x_3$ will have a negative weight. Thus, for each vertex~$u\in V(H)$,  if~$L^{\ell_i,u}$ is the unique list gadget whose center and leaves are in~$\pi_1$, then there are exactly~$\ell_i$ edges incident on~$u$ in~$G$ such that the endpoints of the edge gadgets corresponding to these~$\ell_i$ edges are in~$\pi_1$. Thus, let~$S$ be the set of edges whose center vertices are in~$\pi_1$. Observe that for each vertex~$u\in V(H)$,~$d_{H-S}(u) = \ell_i\in L(u)$ where~$L^{\ell_i,u}$ is the unique list gadget corresponding to~$u$ whose center and leaves are in~$\pi_1$.  Hence,~$(H,L)$ is a Yes-instance of \textsc{General Factors}. 
     
     This establishes that the two instances are equivalent. To complete our proof, finally we establish that treedepth of~$G$ is bounded by a function of the smaller partition of~$H$ (i.e., minimum of~$|X|$ and~$|Y|$). Without loss of generality, let us assume that~$|X|\leq |Y|$. Then, for each~$x\in X$, delete each vertex in~$V^x$ and delete~$s$ (i.e.,~$12|X|+1$ vertices in total). Next, we establish that after deleting these vertices, the treedepth of each component of the remaining graph is bounded. Each list gadget corresponding to vertices of~$x$ is a component of the remaining graph and has treedepth at most~$3$ since deleting the center vertex of a list gadget give components with at most~$2$ vertices each. The remaining components contain the edge gadgets, vertex gadgets and list gadgets corresponding to vertices of~$Y$. Now consider any connencted component in this graph. A component can contain the ($12$) vertices of a vertex gadget corresponding to a vertex~$u$ of~$Y$, the list gadgets corresponding to the vertex~$u$, and the edge gadgets corresponding to edges incident on~$u$. Moreover, no two vertex gadgets are in the same component. Hence, in this component, we can delete at most~$12$ vertices and each component of the remaining graph will either be an edge gadget (a cycle on~$4$ vertices having treedepth at most~$3$) or a list gadget (with treedepth at most~$3$). Thus, the treedepth of~$G$ is at most~$12|X|+1+12+3 = 12|X|+16$. Hence, the treedepth of~$G$ is bounded by a (linear) function of the minimum partition of~$H$. This completes our proof.
\end{proof}

To finalize the complexity picture, in our next result we show that when the valuations are binary, there is an \XP algorithm for the parameterization by the celebrated treewidth. As is common for such result, we employ dynamic programming over a nice tree-decomposition of our graph. However, before we can do so, we settle a relation between a chromatic number of~$G$ and our problem, which allow us to significantly restrict the instance.

\begin{theorem}\label{thm:ASHG:TW:binary:XP}
    When the valuations are binary and the problem is parameterized by the treewidth~$\tw$ of~$G$, it is in \XP to decide whether~$\numCoals$-ASHG~$\Gamma$ admits an individually rational coalition structure. 
\end{theorem}
\begin{proof}[Proof Sketch]
    If~$k \geq \tw+1$, then we can employ the following. 

    \begin{claim}\label{claim_coalitions_more_than_chromatic_number}
        Given a graph~$G$ and a proper coloring of~$G=(V,E)$ with~$c$ colors, we can compute, in polynomial time, a coalition structure with~$k\ge c$ coalitions that is individually rational.
    \end{claim}
    \begin{claimproof}
        Observe first that~$k< |V|$ as otherwise a coalition structure that contains all vertices of~$V$ as singletons is trivially IR. We will transform the proper~$c$-coloring of~$G$ into a proper~$k$-coloring by employing the following procedure. First, we find a color~$i\in[c]$ such that there exist at least two vertices~$u,v$ in~$G$ colored~$i$. This color exists since~$c\le k$ and~$k<|V|$. We then recolor~$u$ by using a new color, and update~$c$ accordingly. We repeat this procedure until~$c=k$. At this stage we have a proper~$k$-coloring of~$G$. We define the coalition structure~$\pi=(\pi_1,\dots,\pi_k)$ such that~$\pi_i$ contains the vertices colored~$i$. This is an IR structure, as each color induces an independent set.
    \end{claimproof}

    Thus, if~$k \ge \tw+1$, then~$(G,k)$ is a \Yes-instance of ASHGs. Furthermore, we can compute an individually rational coalition structure with~$k$ coalitions in \FPT time w.r.t.~$\tw$. To do this it suffices to compute the chromatic number of~$G$. This can be done in \FPT time w.r.t.~$\tw$ using a known algorithm~\citep{AP89}.

    So we may assume that~$k \le \tw$. We solve this case through a bottom-up dynamic programming algorithm on the nice tree-decomposition~$\mathcal{T}$ of the graph~$G$ rooted at a node~$r$. For a node~$t$ of~$\mathcal{T}$, we denote by~$B_t$ the bag of this node and by~$B_t^{\downarrow}$ the set of vertices of the graph that appears in the bags of the nodes of the subtree with~$t$ as a root. Observe that~$B_t\subseteq B_t^{\downarrow}$.

    On a high level, we build a set of partial solutions for each node~$t$ of the tree-decomposition. We then extend these partial solutions so that they are also partial solutions of the parent of~$t$. In particular, we keep coalition structures that may not necessarily be individually rational, but \textit{have the potential} to be in the future. 
    We will also show that it suffices to keep at most~$(tw\cdot \max \{|X_p|,|X_n|\})^{tw}$ representatives of these coalition structures, where~$X_p = \max \{ \sum_{w(e)> 0} w(e) \mid e\in E \}$ and~$X_n = \max \{ -\sum_{w(e)< 0} w(e) \mid e\in E \}$.

    Let~$B_t$ be a bag of~$T$. We denote by~$C_t=(C_i)_{i\in[k]}$ the projection of a coalition structure of~$G$ on~$B_t^\downarrow$, and we say that~$C_t$ is a \textit{partial coalition structure} of~$G$. Also, for any~$u\in B_t$, we denote by~$id(u)$ the index~$i\in[k]$ such that~$u\in C_i$ according to~$C_t$. 
    The information we keep for each node~$t$ of the tree-decomposition is as follows. 
    For each bag~$B_t$ of the tree-decomposition and \textit{partial coalition structure}~$C_t$ we store:
    \begin{itemize}
        \item a table~$I_t$ such that~$I_t[i] = C_i\cap B_t$, for each~$i \in [k]$, and
        \item and a table~$W_t$ such that~$W_t[u]=\sum_{v\in N(u)\cap C_{id[u]}}w(uv)$, for each~$u\in B_t$.
    \end{itemize} 

    Then,~$(G,k)$ is a \Yes-instance of ASHGs if and only if there is a ``partial'' coalition structure of~$B_r$ whose table~$W_r$ contains only non-negative values. 
    
    In order to construct these, we consider each type of node separately. 

    \smallskip
    \noindent\textbf{Leaf Nodes.} Let~$t$ be a leaf node. This case is trivial as~$B_t=\emptyset$. We set~$I_t[i]=\emptyset$. Also, no value of~$W_t$ needs defining.

    \smallskip
    \noindent\textbf{Introduce Nodes.} Let~$t$ be an introduce node,~$u$ be the introduced vertex and~$t'$ be the child of~$t$. For each partial solution~$(W_{t'},I_{t'})$ that we have stored for~$t'$, we guess in which coalition to place~$u$ and update the table~$I_{t'}$ accordingly. We then update~$W_{t'}$: we set~$W_t[u]= \sum_{v\in N(u)\cap C_{id[u]}}w(uv)$. This is correct since~$N(u)\cap C_{id[u]} \subseteq B_t$. Also, for any~$v \in N(u)$ we update~$W_{t'}[v]$ by adding~$w(uv)$. 

    \smallskip
    \noindent\textbf{Forget Nodes.} Let~$t$ be a forget node,~$u$ be the forgotten vertex and~$t'$ be the child of~$t$. This is the step where we make sure that the under-construction coalition structures do indeed have the potential to be IR. 
    For each~$(W_{t'},I_{t'})$ we have stored for~$t'$, we check whether~$W_{t'}[u]\geq 0$. If this isn't the case, we discard the corresponding partial coalition structure. Otherwise, we set~$W_t[v]=W_{t'}[v]$, for all~$v\in B_t$, and update~$I_{t'}$ by removing~$u$ from~$I_{t'}[id[u]]$.

    \smallskip
    \noindent\textbf{Join Nodes.} Let~$t$ be a join node with~$t_1$ and~$t_2$ being its children. We combine the pairs of tables~$(I_{t_1},W_{t_1})$ and~$(I_{t_2},W_{t_2})$ of~$t_1$ and~$t_2$ (respectively) in order to create the pair of tables~$(I_t,W_t)$. We combine two partial coalition structures if they ``agree'' on the vertices of~$B_t$, i.e.,~$I_{t_1}=I_{t_2}$. It then suffices to set the table~$W_t$ appropriately. In particular,~$W_t[u]=W_{t_1}[u]+W_{t_2}[u]-\sum_{v \in N(u)\cap C[id[u]]} w(uv)$. The term with the negative sum is there in order to avoid double counting the edges between vertices of~$B_t$.

    \medskip
    Lastly, observe that two partial coalition structures that agree on the vertices of~$B_t$, for any~$t$, will behave in the same way regarding the final coalition structures of~$G$. 
    Therefore it suffices to consider only one of such partial coalition structures. 
    In total, we need to keep at most~$(tw \cdot \max\{X_p, X_n\})^{\Oh{tw}}$ partial coalition structures per node~$t$.

    As for the running time, it suffices to consider the join nodes, as they are the most computationally demanding types of nodes. For any pair of tables~$(I_{t_1},W_{t_1})$ we consider all pairs of tables~$(I_{t_2},W_{t_2})$ such that~$I_{t_1}=I_{t_2}$. Thus, for any of the~$(tw \cdot\max\{X_p,X_n\})^{\Oh{tw}}$ partial coalition structures of~$t_1$, we consider at most~$(2\cdot\max\{X_p,X_n\})^{tw+1}$ partial coalition structures of~$t_2$. In total, this takes~$(tw \cdot \max\{X_p,X_n\})^{\Oh{tw}}$ time.
\end{proof}

It is not hard to see that, when the valuations are binary, if we additionally parameterize by the maximum degree of~$G$, the previous algorithm becomes fixed-parameter tractable.

\begin{corollary}\label{thm:ASHG:FPT:tw:deg}
    When the valuations are binary and the problem is parameterized by the treewidth of~$G$ and the maximum degree~$\Delta$, it is in \FPT to decide whether a~$\numCoals$-ASHGs~$\Gamma$ admits an individually rational coalition structure.
\end{corollary}

\section{Conclusions}

In our work, we study two natural restrictions of hedonic games, namely~$\numCoals$-hedonic games, where we are asked to form exactly~$\numCoals$ stable coalitions, and hedonic games with size-constrained coalitions, where the goal is to construct~$\numCoals$ coalitions whose sizes are between the given lower and upper bound. We provide a fairly complete algorithmic landscape from the perspective of both classical and parameterized complexity for individual rationality and two restricted variants of hedonic games where valuations are encoded using an underlying preference graph. 

However, we believe that our contribution is much broader. First, we bring to light that individual rationality, a traditionally neglected stability notion in coalition formation, can be surprisingly interesting. For example, in the case of~$\numCoals$-ASHGs, the complexity picture of our problem is identical to the core stability verification problem in (unrestricted) ASHGs~\citep{HanakaKL2024}, even though the techniques used are very different. We believe it is worth revisiting how traditional stability notions interplay. Moreover, to the best of our knowledge, we are the first using the N-fold ILP technique in the area of coalition formation, which turned out to be a very powerful tool in the design of a tractable algorithm for this domain.

An immediate open question left after our work is the complexity classification of ASHGs with fixed size coalitions under binary preferences and bounded treewidth preference graph. Specifically, the algorithm provided in \Cref{thm:ASHG:TW:binary:XP} relies on a connection between ASHGs and the chromatic number of bounded treewidth graphs, which can no longer be easily exploited if we are in the setting of size-constrained coalitions. Another natural direction is to study hedonic games with \emph{friends appreciation} or \emph{enemy aversion} preferences~\citep{LangRRSS2015,OhtaBISY2017,RotheSS2018,ChenCRS2023}, a restrictions of ASHGs which can allow for more tractable cases.

\begin{acks}
    This project was supported by the European Research Council (ERC) under the European Union’s Horizon 2020 research and innovation programme (grant agreement No 101002854), by the European Research Council (ERC) grant titled PARAPATH (grant agreement number 101039913), and co-funded by the European Union under the project Robotics and advanced industrial production (reg. no. CZ.02.01.01/00/22\_008/0004590).
    Foivos Fioravantes is supported by the International Mobility of Researchers MSCA-F-CZ-III at Czech Technical University in Prague,~$\text{CZ}.02.01.01/00/22\_010/0008601$ Programme Johannes Amos Comenius. 
    Nikolaos Melissinos is partially supported by Charles University projects UNCE 24/SCI/008 and PRIMUS 24/SCI/012, and by the project 25-17221S of the Czech Science Foundation (GAČR).
    
    \begin{center}
        \vspace{0.25cm}
        \includegraphics[width=4cm]{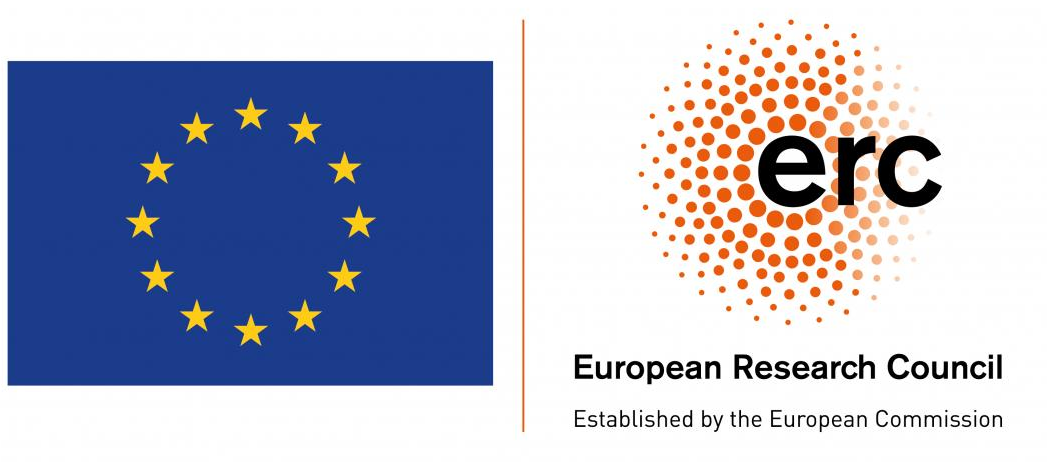}
    \end{center}
\end{acks}

\printbibliography

\appendix

\end{document}